\documentclass[a4paper,UKenglish]{llncs}

\usepackage[utf8x]{inputenc}
\usepackage{pslatex}
\usepackage[OT1]{fontenc}
\usepackage{amsmath}
\usepackage{amsfonts}
\usepackage{amssymb}
\usepackage{wasysym}
\usepackage{stmaryrd}
\usepackage{bussproofs} 
\usepackage{tikz} 
\usetikzlibrary{arrows,shapes,snakes,automata,backgrounds,petri} 
\usepackage{xifthen}
\usepackage{relsize} 

\usepackage{xypic} 

\usepackage{url}

\newcommand{\RLabel}[1]{\RightLabel{\begin{scriptsize} #1 \end{scriptsize}}}

\usepackage[nospace]{cite}

\newcommand{\ERL}{ERL}
\newcommand{\ERLast}{$\mbox{\rm ERL}^*$}

\newcommand{\BI}{BI}
\newcommand{\EL}{EL}
\newcommand{\BBI}{BBI}

\newcommand{\lang}{\mathcal{L}}

\newcommand{\Atom}[1]{\mbox{\rm #1}}
\newcommand{\Prop}{\mbox{\rm Prop}}

\newcommand{\satisfies}{\models_\mathcal{M}}  
\newcommand{\valid}{ \vDash }

\newcommand{\interpretationMRD}[1]{ \ifthenelse{\isempty{#1}}
                                            {\llbracket \cdot \rrbracket}
                                            {\llbracket #1 \rrbracket} }
\newcommand{\interpretationRP}[1]{ \ifthenelse{\isempty{#1}}
                                            {i}
                                            {i(#1)} }

\newcommand{\ARD}[1]{ \ifthenelse{\isempty{#1}}
                                            {ARD}
                                            {ARD(#1)} }

\newcommand{\vp}[1]{{\mathrm #1}}
\newcommand{\BIatop}{\top}
\newcommand{\BIabot}{\bot}
\newcommand{\BIaand}{\wedge}
\newcommand{\BIaor}{\vee}
\newcommand{\BIaneg}{\neg}
\newcommand{\BIaimp}{\rightarrow}
\newcommand{\BImI}{\vp{I}}
\newcommand{\BImast}{\ast}
\newcommand{\BImimp}{\mathbin{-\hspace{-0.1cm}\ast}}

\newcommand{\BIK}[1]{\mathbf{K}_{#1}}
\newcommand{\BIKt}[1]{\widetilde{\mathbf{K}}_{#1}}

\newcommand{\BIC}[2]{\mathbf{L}_{#1}^{#2}}
\newcommand{\BID}[2]{\mathbf{M}_{#1}^{#2}}
\newcommand{\BIE}[2]{\mathbf{N}_{#1}^{#2}}

\newcommand{\BICt}[2]{\widetilde{\mathbf{L}}_{#1}^{#2}}
\newcommand{\BIDt}[2]{\widetilde{\mathbf{M}}_{#1}^{#2}}
\newcommand{\BIEt}[2]{\widetilde{\mathbf{N}}_{#1}^{#2}}

\newcommand{\EK}[1]{\mathbf{K}_{#1}}

\newcommand{\setAg}{ A }
\newcommand{\setR}{ R }
\newcommand{\compR}{ \bullet }
\newcommand{\unitR}{ e }
\newcommand{\relR}[1]{ \sim_{#1} }
\newcommand{\PRM}{ \mathcal{R} }
\newcommand{\model}{ \mathcal{M} }
\newcommand{\val}{ V }
\newcommand{\forcing}[1]{ \vDash_{#1} }
\newcommand{\ESL}{ESL}
\newcommand{\goes}[1]{\stackrel{#1}{\longrightarrow}}

\newcommand{\relCr}{ \simeq } 
\newcommand{\relCag}[1]{  \eqcirc_{#1}} 

\newcommand{\subLabelR}[1]{ \mathcal{E}(#1) }
\newcommand{\domainR}[1]{ \mathcal{D}_r(#1) }

\newcommand{\alphabetR}[1]{ \mathcal{A}_r(#1) }
\newcommand{\closure}[1]{ \overline{#1} }

\newcommand{\labelledFormula}[3]{ ( \mathbb{#1} #2 : #3 )}
\newcommand{\CSS}[2]{ \langle #1, #2 \rangle }
\newcommand{\CSSfiniteInclusion}{ \preccurlyeq_f }
\newcommand{\CSSinclusion}{ \preccurlyeq }
\newcommand{\concatList}{ \oplus }

\newcommand{\realizationR}[1]{ \vert #1 \vert }

\newcommand{\lgg}{ \llbracket } 
\newcommand{\ldd}{ \rrbracket } 

\newcommand{\satisfaction}[1]{ \models_{#1}}

\newcommand{\setW}{W}
\newcommand{\relSymbol}{ \mathbf{R} }
\newcommand{\valuation}{ \mathcal{V} }

\newcommand{\BIdp}[1]{\lozenge_{#1}}
\newcommand{\BIbp}[1]{\square_{#1}}
\newcommand{\rel}[4]{ (#1, #2) \relSymbol (#3, #4) }
\newcommand{\BId}{\lozenge}
\newcommand{\BIb}{\square}
\newcommand{\BIdr}{\lozenge_{\compR}}
\newcommand{\BIbr}{\square_{\compR}}

\title{A Substructural Epistemic Resource Logic: \\ Theory and
  Modelling Applications} 
\author{Didier Galmiche*, Pierre Kimmel*, and David Pym$\dagger$}
\institute{*Universit\'e de Lorraine, CNRS, LORIA, France \quad
  $\dagger$University College London, UK}

\begin{document}

\maketitle

\begin{abstract}
We present a substructural epistemic logic, based on Boolean BI,
in which the epistemic modalities are parametrized on agents' local
resources. The new modalities can be seen as generalizations
of the usual epistemic modalities. The logic combines Boolean
BI's resource semantics --- we introduce BI and its resource semantics 
at some length --- with epistemic agency. We illustrate the use
of the logic in systems modelling by discussing some examples 
about access control, including semaphores, using resource tokens. 
We also give a labelled tableaux calculus and establish soundness 
and completeness with respect to the resource semantics.
\end{abstract}

\section{Introduction} \label{sec:introduction}

The concept of resource is important in many fields including, among
others, computer science, economics, and security. For example, in 
operating systems, processes access system resources such as memory, 
files, processor time, and bandwidth, with correct resource usage 
being essential for the robust function of the system. The internet 
can be regarded as a giant, dynamic net of resources, in which Uniform 
Resource Locators refer to located data and code. 

In recent years, the concept of resource has been studied and analysed
in computer science through the bunched logic, BI,
\cite{OP1999,POY,Gal05a} and its variants,  such as Boolean \BI\
(\BBI) \cite{OHea01a} and bunched modal logics \cite{Cour15a,CGP15},
and applications, such as Separation Logic \cite{OHea01a,Rey02a}.   

The truth-functional, Kripke semantics of these logics, based on
preordered partial monoids is sketched below. However, before
proceeding to describe this  semantics, it is perhaps worth observing
that this choice of structure for BI's  models can be motivated
directly in terms of natural requirements for the  properties of a
notion of resource. Assuming a set of resource elements, we  expect to be
able to  
\begin{itemize}
\item[-] combine two resource elements to give a new resource element,
	and 
\item[-] to be able to compare two resource elements, to determine
	which is the greater.
\end{itemize}  
It is also natural to expect that the combination of elements be
partial and this is indeed amply justified by leading examples. These
simple assumptions, that  around are cleanly captured by preordered
partial monoids, have led to a remarkably useful `resource
semantics'. The need for partiality arises in two ways. Conceptually,
we observe that in our semantics of resources it is quite natural to
expect that not all combinations of resource elements will exist
(Separation Logic \cite{OHea01a,Rey02a} provides an immediate and
compelling example). Second, partiality is technically convenient for
BI's metatheory \cite{Gal05a}.  

These considerations lead to a semantics for BI based on 
partially ordered partial monoids of worlds,
\[
	\mbox{{\bf R}} = (R , \sqsubseteq , \bullet , e). 
\] 
Here, composition of resources is captured by the partial 
monoidal operation, $\bullet$, with unit $e$, and comparison 
of resources is captured by the partial order $\sqsubseteq$. 
Where defined, this structure is required to satisfy the 
bifunctoriality condition that if $r_1 \sqsubseteq s_1$ and 
$r_2 \sqsubseteq s_2$, then $r_1 \bullet r_2 \sqsubseteq s_1 \bullet
s_2$.  Let us note that $\downarrow$ denotes definedness of the
composition. 
 
Given such structures, the logic BI of bunched implications --- see,
for example, \cite{OP1999,Pym02,POY,Gal05a} --- which freely combines
intuitionistic propositional additives with intuitionistic
propositional multiplicatives --- has its Kripke semantics given by
the following  satisfaction relation, where $V$ is an interpretation
of propositional letters in $\wp(R)$, in the usual way: 
\[
\begin{array}{rcl}
r \models \Atom{p} & \mbox{\rm iff} & \mbox{$r \in V(\Atom{p})$}  \\
r \models \BIabot   & \mbox{\rm never} &  \\
r \models\BIatop    & \mbox{\rm always} &  \\
r \models \BIaneg \phi  & \mbox{\rm iff} & \mbox{\rm $r \not\models \phi$} \\
\end{array} \qquad
\begin{array}{rcl}
r \models \phi \BIaor \psi     & \mbox{\rm iff} & 
	\mbox{\rm $r\models \phi$ or $r\models \psi$} \\
r \models \phi \BIaand \psi  & \mbox{\rm iff}  & 
	\mbox{\rm $r\models \phi$ and $r\models \psi$} \\
r \models \phi \BIaimp \psi  & \mbox{\rm iff} & \mbox{for all $r
                                                \sqsubseteq s$,} \\  
	& & \mbox{\rm if $s \models \phi$, then $s\models \psi$} \\
\end{array}
\]
\[
\begin{array}{rcl}
r \models \BImI & \mbox{\rm iff} & \mbox{\rm $e \sqsubseteq r$} \\
r \models \phi \BImast \psi  & \mbox{\rm iff} & \mbox{\rm there exist 
	$r_1 , r_2 \in R$ s.t.     $r_1\bullet r_2\downarrow$, $r
                                                \sqsubseteq r_1\bullet
                                                r_2$, and} \\  
	& & \mbox{\rm $r_1\models \phi$ and $r_2 \models \psi$} \\
r \models \phi \BImimp \psi & \mbox{\rm iff} & \mbox{\rm for all $r' \in R$, if 
	$r \bullet r' \downarrow$ and $r' \models \phi$,}  \\
	& & \mbox{\rm then $r \bullet r' \models \psi$}
\end{array}
\]

This resource semantics for BI --- that is, the interpretation of BI's
semantics in terms of  resources --- underpins its applications to
Separation Logic --- and its family of derivatives;  see
\cite{DP-MFPS-2017,DP-2019-LMCS} for an extensive discussion --- and
is mainly  concerned with sharing and separation. 

Specifically, Separation Logic is usually given as a presentation
(often using Hoare triples) of a \emph{specific theory of Boolean BI}
for a language of memory  cells and pointers with a model based on the
stack and the heap \cite{OHea01a}.  Versions of Separation Logic that
are based on (intuitionistic) BI, as given  above, are also possible
\cite{OHea01a}.   

In Boolean BI (BBI), \cite{OHea01a,Rey02a}, the additives are
classical, so that the order  is collapsed to equality in the partial
monoid. Thus we have   
\[
\begin{array}{rcl}
r \models \phi \BIaimp \psi  & \mbox{\rm iff} & \mbox{if $r \models
                                                \phi$, then $r \models
                                                \psi$} \\ 
r \models \BImI                   & \mbox{\rm iff} & \mbox{\rm $e = r$} \\
r \models \phi \BImast \psi  & \mbox{\rm iff} & \mbox{\rm there exist 
	$r_1 , r_2 \in R$ s.t.     $r_1\bullet r_2 \downarrow$, $r =
                                                r_1\bullet r_2$, and}
  \\  
	& & \mbox{\rm $r_1\models \phi$ and $r_2 \models \psi$} 
\end{array}
\]
The semantics described above is otherwise unchanged.  

Thus sharing of resources is captured by additive connectives, such as
$\BIaand$, while  separation of resources is captured by
multiplicative connectives, such as $\BImast$.  These connectives are
the logical kernels of the family of separation logics, with
resources being interpreted in various ways, such as memory regions,
\cite{OHea01a,Rey02a}, or elements of other particular monoids of
resources  \cite{CMP12}. This semantic view of resource stands in
stark contrast to the the `number-of-uses' reading of Linear Logic's
proof theory \cite{Gir86}. We shall  return to this point in the
sequel, where we consider the evolution of a model of  system of
resources.   

This framework of resource semantics has also been extended into modal
logic.  Specifically, we can set up a conservative extension (a `Logic
of Separating Modalities'  or LSM \cite{CGP15})  of the modal logic S4 
which adds \emph{multiplicative modalities} --- modalities that are
parametrized on (local) resources. These modalities are defined
relative to two-dimensional  worlds, one of which captures the S4
accessibility relation and one of which supports the  resource
parametrization.  
        
Roughly speaking, an LSM model is a 4-tuple $(\setW, \PRM, \relSymbol,
\valuation)$, where $W$ is a set of worlds, $\mathcal{R}$ is a partial
monoid of `resources' $(Res, \bullet, e)$, ${\bf R} \subseteq (W
\times Res) \times (W \times Res)$ is a reflexive and transitive
relation, and $V$ is an interpretation of propositional letters  in
$\wp(W \times Res)$.  Then, using the both dimensions of `worlds' to
handle, respectively, both classical modality and resource
parametrization,  we have  
\[
\begin{array}{rcl}
w, r \models \BIdp{s} \phi & \mbox{\textup{iff}} & \mbox{\textup{there
                                                   exist $w' \in
                                                   \setW$ and $r' \in
                                                   \setR$ such that 
     $r \compR s \downarrow$,}} \\
& & \mbox{\textup{$\rel{w}{r \compR s}{w'}{r'}$ and $w', r' \models \phi$}} \\
          & & \\
w, r \models \BIbp{s} \phi & \mbox{\textup{iff}} & \mbox{\textup{for
                                                   all  $w' \in \setW$
                                                   and all  $r' \in
                                                   \setR$, if $r
                                                   \compR s 
  \downarrow$ and}} \\
&  & \mbox{\textup{$\rel{w}{r \compR s}{w'}{r'}$, then $w', r' \models
     \phi$.}} \\ 
\end{array}
\]
Here, $s$ is the local resource, associated with the modality, and
$r$, in the model, is the ambient resource. The modalities are read as
asserting that $\phi$ is possibly (respectively, necessarily) true at
the world $(w,r)$ subject to the availability of additional resource $s$. 

Note that two other pairs of modalities are derivable from these:
\begin{itemize}
\item[-] The basic additive modalities: 
\[
\begin{array}{rcl}
w, r \models \BId \phi & \mbox{\textup{iff}} & \mbox{\textup{there
                                               exist $w' \in \setW$
                                               and $r' \in \setR$ such
                                               that $\rel{w}{r}{w'}{r'}$}} \\
                             &
         & \mbox{\textup{and $w',r' \models \phi$}} \\
w, r \models \BIb \phi & \mbox{\textup{iff}} & \mbox{\textup{for
  all $w' \in \setW$ and all $r' \in \setR$,  if $\rel{w}{r}{w'}{r'}$
  then}} \\
                      &      & \mbox{\textup{$w', r'  \models \phi$.}} \\
\end{array}
\]
\item[-] Multiplicative modalities with undetermined additional resource parameters: 
\[
\begin{array}{rcl}
w, r \models \BIdr \phi & \mbox{\textup{iff}} & \mbox{\textup{there exist $w' \in \setW$ and $s, r' \in \setR$ such that  $r
  	\compR s \downarrow$,}} \\
                                         &
                                                         &  \mbox{\textup{$\rel{w}{r
                                                             \compR
                                                             s}{w'}{r'}$, and $w', r' \satisfaction{\model} \phi$}} \\
w, r \models \BIbr \phi & \mbox{\textup{iff}} & \mbox{\textup{for
  	all $w' \in \setW$ and all $s, r' \in \setR$, if                                $(r \compR s \downarrow$ and}} \\
                                   &
                                                       &  \mbox{\textup{$\rel{w}{r
                                                           \compR
                                                           s}{w'}{r'})$ then $w', r' \models \phi$.}} \\
\end{array}
\]
\end{itemize}
Full details of the derivations of these modalities may be found in
\cite{CGP15} (Lemma 6), where the conservativity of LSM over S4 
is also established (in Section 5). The key feature of BI as a
modelling tool (and hence of its specific model Separation Logic) is
its control of the representation and handling of resources provided
by the resource semantics and the associated proof systems. Notice
that, in the semantics given above, the components of the additive
conjunction, $\wedge$, share resources whereas the truth condition for
the multiplicative conjunction, $\ast$, requires separate resources
for each component. Notice also that this interpretation extends to
the multiplicative implication as follows: $\BImimp$ can be seen as
(the type of) a function that combines the resource required to support  
itself with the resource required to support its argument to give the 
resource required to support the application of the function to its argument 
(see \cite{OP1999,OHearn2003}). Finally, notice also that we do \emph{not} 
assume (in the manner of hybrid logic) the existence
of an atomic proposition for each element `s' of the set $Res$ with $r
\models s$ iff $r = s$: from the perspective of  resource semantics,
such an assumption --- the motivations for which would be somewhat
technical and essentially syntactic --- is not well supported. In
particular, we would argue that such an assumption obscures the
natural structure of the modalities that we wish to  explore
and. moreover, imposes a constraint on the relationship between worlds
and their  properties that we do not wish to take in general. We will
return to this point briefly in  Section~\ref{sec:erl}. 

BI's sequent proof systems employ \emph{bunches}, with two context-building 
operations: one for the additives ---characterized by $\wedge$, which
admits weakening and contraction --- and one for the multiplicatives
--- characterized by $\ast$, which admits neither weakening nor
contraction. Bunches are not finite sequences of formulae, but rather
are finite trees, with formulae at the leaves and the context building
operations at the internal vertices.  For the details of the set-up,
see \cite{OP1999,POY,OHearn2003}. 

In this set-up, we have the following right rules for the conjunctions
and their corresponding implications, $\rightarrow$ and $\BImimp$: 
\[
\frac{\Gamma \vdash \phi \quad \Delta \vdash \psi}
	{\Gamma \, ; \, \Delta \vdash \phi \wedge \psi}\quad\mbox{$\wedge$R}
\qquad \mbox{\rm and} \qquad 
\frac{\Gamma \, ; \, \phi \vdash \psi}{\Gamma \vdash \phi \rightarrow
  \psi}\quad\mbox{$\rightarrow$R}  
\]
and 
\[
\frac{\Gamma \vdash \phi \quad \Delta \vdash \psi}
	{\Gamma \, , \, \Delta \vdash \phi \ast \psi}\quad\mbox{$\ast$R}
\qquad \mbox{\rm and} \qquad 
\frac{\Gamma \, , \, \phi \vdash \psi}
	{\Gamma \vdash \phi \BImimp \psi}\quad\mbox{$\BImimp$R}.
\]
Again, details may be found in the references given above. 

In this setting, the structural rules of Weakening and Contraction arise as follows:
\[
	\frac{\Gamma(\phi) \vdash \chi}{\Gamma(\phi \, ; \, \psi)
          \vdash \chi}\quad\mbox{W} \qquad\mbox{\rm and}\qquad 
	\frac{\Gamma(\phi \, ; \, \phi) \vdash \psi}{\Gamma(\phi)
          \vdash \psi}\quad\mbox{C}.  
\]
In the former rule, the leaf $\phi$ is replaced by the bunch $\phi \,
; \, \psi$ and, in the latter rule,  the sub-bunch (in the evident
sense) $\phi \ ; \, \phi$ is replaced by the formula $\phi$.  In both
cases, \, ; \, (rather than ,~) is used. Again, details may be found
in the references given above.   

The soundness and completeness of BI's proof systems for the semantics
given above is established in \cite{OP1999,POY} and elsewhere and via
labelled tableaux in \cite{Gal05a}, and the completeness of BBI for
the partial monoid semantics described above is discussed 
comprehensively in \cite{Lar14a}.   

The idea of resource semantics as it derives from BI and its models and 
its use as modelling tool is discussed extensively in \cite{Pym2019}, in 
an article that is intended to be widely accessible to logicians and 
computer scientists. 

Girard's Linear Logic (LL) \cite{Gir86} also decomposes the logical
connectives into additive and multiplicative forms (for classical and
intuitionistic conjunction and disjunction, but not for intuitionistic
implication). However, it does so in a very different way from
BI. Instead of employing bunches to allow control of the structural
rules, LL introduces the so-called exponentials ! and ? --- modalities, 
similar to S4's $\Box$ and $\Diamond$) --- which have the
following left and right rules:  
\[
\begin{array}{c@{\qquad}c}
\dfrac{\Gamma , \phi \vdash \Delta}{\Gamma , ! \phi \vdash \Delta}\quad{!L} & 
	\dfrac{! \Gamma \vdash \phi , ?\Delta}{! \Gamma \vdash ! \phi , ?\Delta}\quad{!R}
  \\ & \\ 
\dfrac{! \Gamma , \phi \vdash ? \Delta}{! \Gamma , ?\phi \vdash ? \Delta}\quad{?L} & 
	\dfrac{\Gamma \vdash \phi , \Delta}{\Gamma \vdash  ? \phi , \Delta}\quad{? R} \\
\end{array}
\]
Then the structural rules of Weakening and Contraction arise as
\[
\begin{array}{c@{\qquad}c}
\dfrac{\Gamma \vdash \Delta}{\Gamma , ! \phi \vdash \Delta}\quad{W L} & 
	\dfrac{\Gamma \vdash \Delta}{\Gamma \vdash ? \phi , \Delta}\quad{W R} 
\end{array}
\] 
and 
\[
\begin{array}{c@{\qquad}c}
\dfrac{\Gamma , ! \phi , ! \phi \vdash \Delta}{\Gamma , ! \phi \vdash \Delta}\quad{C L} & 
	\dfrac{\Gamma \vdash ? \phi , ? \phi , \Delta}{\Gamma \vdash ? \phi ,  \Delta}\quad{C R} 
\end{array}
\]
Restricting to a single-conclusioned calculus for intuitionistic LL,
we have just the $!$\,.  

At this point, we may ask what is the relationship between BI and
LL. The short answer is that they are essentially incomparable. This
is explained in detail in the references  given above (e.g.,
\cite{OP1999,Pym02,Pym2019}), but the essential point can be seen  in
terms of their differing treatments of intuitionistic implication. In
BI, which can be  considered to  freely combines intuitionistic
propositional logic and multiplicative  propositional linear logic,
intuitionistic implication is present directly. In LL, intuitionistic
implication, $\phi \supset \psi$, is represented using Girard's
translation 

\begin{equation} \label{eqn:!-endo}
	\phi \supset \psi \,=\, !\, \phi \multimap \psi 
\end{equation}

Such a representation does not exist in BI. This can be seen, as
described in \cite{OP1999,Pym02,Pym2019}, using an argument based on
category-theoretic models of BI's proofs.  Specifically, BI's proofs
are modelled by bi-cartesian doubly closed categories, and there is
no endofunctor $!$ on such a category that satisfies (the
interpretation of) Equation (\ref{eqn:!-endo}).  

Returning briefly to truth-functional semantics and its resource
interpretation, we remark that LL's recently developed Kripke
semantics \cite{Coumans2014} does not, as it  stands, admit a direct
resource interpretation of the kind outline above. The possibility
of such interpretations is an interesting issue.  

Modal extensions of BI, such as MBI \cite{CMP12,AP16}, DBI, and DMBI
\cite{Cour15a}, have been proposed to introduce dynamics into resource
semantics. In recent work, the idea of introducing agents, together
with their knowledge, into the resource semantics has led to an
Epistemic Separation Logic, called ESL, in which epistemic possible
worlds are considered as resources \cite{Cour15b}. This logic
corresponds to an extension of Boolean BI with a knowledge modality,
$\BIK{a}$, such that $\BIK{a} \phi$ means that the agent $a$ knows
that $\phi$ holds.   

Various previous works on epistemic logics consider the concept of
resource, using a variety of approaches. They include
\cite{Baltag06,Halpern03,Naumov15}. Here we aim to explore more deeply
the idea of epistemic reasoning \cite{EpiHandbook} in the context of
resource semantics, and its associated logic, by taking the basic
epistemic modality $\BIK{a}$ and parametrizing it with a resource $s$,
with the associated  introduction of relations not only between
resources,  according to an agent, but also between composition  of
resources in  different ways. The parametrizing resource may be
thought of as being associated with, or local to, the agent. This
approach leads to the definition of two new modalities $\BIC{a}{s}$
and $\BID{a}{s}$, and, consequently, to a new logic in which, as a
leading example, we can obtain an account of access to resources and
its control, whether they be  pieces of knowledge, locations, or
other entities. We call this logic \emph{Epistemic Resource Logic} or
\ERL.   

In Section~\ref{sec:erl}, we set up the logic \ERL\ by a semantic
definition and, in Section~\ref{sec:properties},  we give the key
conservative extension properties   of the logic and also introduce a
useful sublogic, \ERLast. In Section~\ref{sec:modelling},  we explain
how to  use the logic to model and reason about the relationship
between a security policy --- in the  context of access control ---
and the system  to which it is applied (cf.  Schneier's Gate problem
\cite{Schneier2005}). Our application to systems security policy
stands in contrast  to other work (e.g., \cite{Pucella2015}) in which
epistemic logic  has been applied to the analysis of cryptographic
protocols. We complete this section with other examples, including
joint access and semaphores, which illustrate the applicability of \ERL\
in these perspectives. In Section~\ref{sec:tableaux}, we set up a
labelled tableaux calculus for \ERL, and establish soundness with
respect to \ERL's semantic definition and also completeness from a
countermodel extraction method. Let us note that we apply the approach
and techniques already used for designing such labelled tableaux for
other modal extensions of \BBI\ \cite{Cour15a,CGP15,Cour15b}. Details
of the arguments are provided  in the appendices. Our arguments
encompass also the sublogic \ERLast.   

Further work will be devoted to further study of the logic and its
variants, including intuitionistic and dynamic systems, to local
reasoning  for resource-carrying agents  \cite{OHea01a,Rey02a}, to
connections with other approaches to modelling the relationship
between policy and implementation in system management
\cite{Caires10}, and to approaches  involving logics for layered
graphs \cite{AP16,CMP15}.  The work presented here builds upon and 
strongly develops early ideas presented in \cite{ICLA17}.

\section{An epistemic resource logic} \label{sec:erl} 

Epistemic logic is the logic of knowledge and belief. It is concerned
with what agents know and believe. The knowledge and beliefs of agents
are represented  using modalities which assert the truth of
propositions relative to agents' judgements of the relationship
between worlds \cite{EpiHandbook}. In the setting of resource
semantics, worlds are interpreted as representing available resources
and agents make judgements about the equivalence of resources.  

The language $\lang$ of the epistemic resource logic, or \ERL, is
obtained by adding two new modal operators $\BIC{}{}$ and $\BID{}{}$
to the \BI\ language. In order to define the language of \ERL, we
introduce the following structures: a finite set of agents $A$; a
finite set of resources $Res$, with a particular element, $e$; an
internal composition operator $\cdot$ on $Res$ ($\cdot : Res \times
Res\rightharpoonup Res$); a countable set of propositional symbols
$\Prop$. The language $\lang$ of \ERL\ is defined as follows: 
\[
\phi ::= \Atom{p} \mid \BIabot \mid \BIatop \mid \neg \phi \mid
   \BImI \mid \phi \BIaor
   \psi \mid \phi \BIaand \psi \mid \phi \BIaimp \phi \mid \phi \BImast
   \phi \mid \phi \BImimp \phi \mid \BIC{a}{s}\phi \mid \BID{a}{s} \phi, 
\]
where $\Atom{p} \in \Prop$, $a\in A$ and $s\in Res$. 

In this context we call $s$ the agent's {\em local resource}. We also
define the following operators:  
$\BIDt{a}{s} \phi \equiv \BIaneg \BID{a}{s} \BIaneg\phi$ and
$\BICt{a}{s} \phi \equiv \BIaneg \BIC{a}{s} \BIaneg\phi$. The meanings
of these connectives are defined in the sequence of definitions that
follow below.  For simplicity, we write $rs$ instead of $r \cdot s$
and so write $\BIC{a}{rs}\phi$ instead of $\BIC{a}{r \cdot s} \phi$.  

Note that we introduce modalities that depend on agents and resources,
and compare them with previous work on an epistemic extension of
Boolean BI \cite{Cour15b}. With a slight abuse of
notation, we have explicit resources  in the language syntax: just
as in \cite{CGP15}, we must assume  that the resource elements present
in the syntax of the modalities  have counterparts in the partial
resource monoid semantics. This design choice has consequences both
for  the expressivity of the logic and for the formulation of the
tableaux calculus.  In the sequel, $\downarrow$ denotes definedness and 
$\uparrow$ undefinedness. 

\begin{definition}[Partial resource monoid] \label{def_prm} 
A \emph{partial resource monoid} (PRM) is a structure
$\mathcal{R}=(R,\bullet)$ such that 
\begin{itemize}
\item $R$ is a set of \emph{resources} such that $Res\subseteq R$
	(which notably means that $e\in R$), and 
\item $\bullet : R \times R \rightharpoonup R$ is an operator on $R$ such
	that, for all $r_1,r_2,r_3 \in R$, 
	\begin{itemize}
	\item[-] $\bullet$ is an extension of $\cdot$: if $r_1, r_2, r_3 \in
	Res$, then $r_1 = r_2 \cdot r_3$ iff $r_1 = r_2 \bullet r_3$, 
	\item[-] $e$ is a neutral element: $r_1 \bullet e\downarrow$ and
	$r_1 \bullet e=r_1$, 
	\item[-] $\bullet$ is commutative: if $r_1 \bullet r_2\downarrow$, then
	$r_2\bullet r_1\downarrow$ and $r_2\bullet r_1 = r_1\bullet r_2$, and 
	\item[-] $\bullet$ is associative: if $r_1 \bullet (r_2 \bullet
 	 r_3)\downarrow$, then $(r_1 \bullet r_2) \bullet r_3\downarrow$ and
 	 \\ $(r_1 \bullet r_2\bullet)  r_3 = r_1 \bullet (r_2 \bullet r_3)$.  
	\end{itemize}
\end{itemize}
\end{definition}
We call $e$ the \emph{unit resource}  and $\bullet$ the \emph{resource
  composition}. Henceforth, $\wp(R)$ denotes the powerset of $R$. 
  
Note that we implicitly consider that the resource composition
$\bullet$ is compatible with equality between resources. That means
that if $r_1=r_2$ and $r_1 \bullet r_3 \downarrow$, then $r_2 \bullet
r_3\downarrow$ and $r_2 \bullet r_3 = r_1 \bullet r_3$
(right-composition  property of $\bullet$). We also have the
left-composition since $\bullet$ is commutative.  

\begin{definition}[Model] \label{def_model}
A \emph{model} is a triple $\mathcal{M}=(\mathcal{R},\{\sim_a\}_{a\in
  A}, V)$ such that 
\begin{itemize}
	\item $\mathcal{R} = (R , \bullet)$ is a PRM, 
	\item for all $a \in A$, $\sim_a\subseteq R\times R$ is an
          equivalence relation, and 
	\item $V : \Prop \rightarrow \wp(R)$ is a valuation function.
\end{itemize}
\end{definition}

We can place this logic in the context of our previous work on
modal \cite{CMP12,CMP15}  and epistemic extensions of (Boolean)
BI \cite{Cour15a,Cour15b}.  In \cite{Cour15b}, an epistemic
extension of Boolean BI, called ESL, is introduced.  In this logic,
there is just one epistemic modality, $K_a$, which allows the
knowledge  of an agent $a$ to be expressed. The modalities employed in
this system and those employed in the system presented herein stand in
contrast to the  modalities of the system LSM described in
Section~\ref{sec:introduction}  in that they make essential use of the
notion of agent in their definition.  

More formally, the semantics of this modality is  defined by $r
\satisfies K_a\phi$ if and only if, for all $r'$ such that $r \sim_a
r'$,  $r' \satisfies\phi$, where $r$ and $r'$ are semantic worlds (or
resources) and $\sim_a$ is a relation between worlds that expresses
that they are equivalent from the point of view of the agent $a$. The
parametrization of modalities on resources derives from ideas that are
conveniently  expressed in, for example, \cite{CMP12,CMP15}.  \\

In this paper, we aim to develop the idea in order to consider a modality
like $K_a$ and to parametrize it on a resource $s$, requiring the
world relation to be of the form  $r \bullet s \sim_a r'$ or $r \sim_a
r' \bullet s$ or even $r \bullet s \sim_a r'\bullet s$.  Then, in the
spirit of \ESL, we define a new logic from Boolean BI that allows us
to model not only relations between resources according to an agent,
but also how those relations are restricted by resources. We can also 
consider the resources upon which the agent's relation are parametrized 
to be local to the agent. 

In this spirit, we define two new modalities $\BIC{a}{s}\phi$ and 
$\BID{a}{s}\phi$, with the notation building on the usual one in epistemic logic, 
for which we have the following semantics expressing two forms of the 
agent's contingency for truth in the presence of composable resources:  

\begin{enumerate}
\item $\BIC{a}{s}\phi$ expresses that the agent, $a$, can establish
 the truth of $\phi$ using a given resource whenever the ambient
 resource, $r$, can be combined with the  agent's local resource,
 $s$, to yield a resource that $a$ judges to be  equivalent to that
 given resource. 

\medskip 
 
In other words $\BIC{a}{s}\phi$ is true relative to the ambient
resource, $r$, iff for $a$'s views of the combination 
of the ambient resource, $r$, and its local resource, $s$, $\phi$ is
true. More formally we have  
\[
	\begin{array}{rcl}
	r \satisfies \BIC{a}{s}\phi & \mbox{\rm iff} & \mbox{\rm if $ r
	\bullet s \downarrow$ then for all $r'\in R$, if
	$r \bullet s \sim_a r'$, then $r' \satisfies\phi$}
	\end{array}
\]

\item $\BID{a}{s}\phi$ expresses that the agent, $a$, can establish
 the truth of $\phi$ if there exists a resource that can be combined
 with  its local resource, $s$, such that  $a$ judges the  combined
 resource  to be equivalent to the ambient resource, $r$.   
 
 \medskip 
 
In other words $\BID{a}{s}\phi$ is true relative to the ambient
resource, $r$, iff for $a$'s views, the ambient resource is the
combination of the local resource, $s$, with another resource that
makes $\phi$ true. More formally we have  
\[
\begin{array}{rcl}
r \satisfies \BID{a}{s}\phi & \mbox{\rm iff} & \mbox{\rm there
exists $r' \in R$ such that $r' \bullet s \downarrow$ and 
$r \sim_a r' \bullet s$ and $r'  \bullet s \satisfies \phi$}
\end{array}
\]
\end{enumerate}

\ERL\ can thus be seen as a particular epistemic logic that provides
new modalities which model access to resources, whether they are
interpreted as pieces of knowledge, locations, or otherwise. \\

Note that we could obtain operators with similar semantics by taking 
the epistemic separation logic \ESL\ \cite{Cour15b} and adding it the 
hybrid operators of the hybrid logic HyBBI \cite{Broth14}.
Such a new logical framework would allow us to use symbols, called
nominals, that force a  formula to be valid for a specific
resource. Namely, if we consider a nominal $n_s$ forcing the resource
$s$, we then could define the modality $\BIC{a}{s}\phi$ by
$\BIC{a}{s}\phi \equiv n_s \BImimp \BIK{a}\phi$ and we recover the
semantics given in this section for this modality. Moreover, we could
also define the modality $\BID{a}{s}\phi$ by $\BID{a}{s}\phi \equiv
\BIKt{a}  ((\top \BImast n_s) \BIaand \phi)$. Observations  like this
are quite common for logics of the kinds considered heren but our view
is that conceptual clarity, rather than syntactic ingenuity, should
drive the design choices.  

This hybrid approach based on nominals represents a significant
technical  addition to our semantic assumptions that is not justified
by the motivations of resource semantics, adding a confusion between 
resources and propositions that we consider to be inconvenient for our
intended modelling applications. Moreover, we would argue that the
identities between the modalities that are induced obscures rather
than elucidates their meaning --- although we would concede that the
identities may be of use in mechanical implementations --- and leads
to a less elegant analysis. Furthermore, working with the hybrid
semantics requires additional work in setting the tableaux-based
metatheory for the logic, as discussed in Section~\ref{sec:tableaux}. 

It therefore seems appropriate to add the epistemic operators
systematically in a clean semantic setting. 

\begin{definition}[Satisfaction and validity] \label{def:sat-valid}
Let  $\mathcal{M}=(\mathcal{R},\{\sim_a\}_{a \in A}, V)$ be a
model. The satisfaction relation $\satisfies \subseteq R\times\lang$
is defined, for all $r \in R$, as follows:
\[
{ \begin{array}{rcl}
r \satisfies \Atom{p}               & \mbox{\rm iff} &
\mbox{$r \in V(\Atom{p})$}  \\
r \satisfies \BIabot                 & \mbox{\rm never} &  \\
r \satisfies\BIatop                  & \mbox{\rm always} &  \\
r \satisfies \BIaneg \phi         & \mbox{\rm iff} &
\mbox{\rm $r\not\satisfies \phi$} \\
\end{array} \qquad
\begin{array}{rcl}
r \satisfies \phi \BIaor \psi     & \, \mbox{\rm iff} \, &
	\mbox{\rm $r\satisfies\phi$ or $r\satisfies\psi$} \\
r \satisfies \phi \BIaand \psi  & \, \mbox{\rm iff} \, & \mbox{\rm
 	$r\satisfies\phi$ and $r\satisfies\psi$} \\
r \satisfies \phi \BIaimp \psi  & \, \mbox{\rm iff} \, & \mbox{\rm if $r
 \satisfies \phi$, then $r\satisfies\psi$} \\
\end{array}}
\]
\[
{\begin{array}{rcl}
r \satisfies \BImI                   & \mbox{\rm iff} &
\mbox{\rm $r = e$} \\
r \satisfies \phi \BImast \psi  & \mbox{\rm iff} & \mbox{\rm there
 exist $r_1 , r_2 \in R$ s.t.     $r_1\bullet r_2\downarrow$,
 $r_1\bullet r_2 = r$, and } \mbox{\rm $r_1\satisfies \phi$ and $r_2
 \satisfies \psi$} \\
r \satisfies \phi \BImimp \psi & \mbox{\rm iff} & \mbox{\rm for all
 $r' \in R$, if $r\bullet
   r' \downarrow$ and $r' \satisfies\phi$, then $r \bullet r'
       \satisfies \psi$} \\
   & & \\
r \satisfies \BIC{a}{s}\phi & \mbox{\rm iff} & \mbox{\rm if $ r \bullet s
                                              \downarrow$ then for all
 $r'\in R$, if $r \bullet s \sim_a r'$, then $r' \satisfies\phi$} \\
r \satisfies \BID{a}{s}\phi & \mbox{\rm iff} & \mbox{\rm there
exists $r' \in R$ such that $r' \bullet s \downarrow$ and 
$r \sim_a r' \bullet s$ and $r'  \bullet s \satisfies \phi$.} \\
\end{array}}
\]
A formula $\phi$ is \emph{valid}, denoted $\valid\phi$, if and only if,
for any model $\mathcal{W}$ and any resource $r$, we have $r
\satisfies \phi$.
\end{definition}
 
\begin{proposition}[Satisfaction for the secondary modalities] 
	\label{prop:sat-secondary-mods} 
Let  $\mathcal{M}=(\mathcal{R},\{\sim_a\}_{a\in A}, V)$ be a model,
and let $r\in R$. The following statements hold: 
\begin{enumerate}
\item $r \satisfies \BICt{a}{s} \phi$ iff if $r \bullet s \downarrow$
  then there exists $r' \in R$ such that  $r \bullet s \sim_a r'$ and
  $r' \satisfies \phi$;   
\item $r \satisfies \BIDt{a}{s}\phi$ iff for all $r'\in R$, if $r'
\bullet s  \downarrow$ and $r\sim_a r' \bullet s$,  then $r' \bullet s
\satisfies\phi$. 
\end{enumerate}
\end{proposition}
\begin{proof}

Consider the first part, 1. 
$\BICt{a}{s} \phi \equiv \BIaneg \BIC{a}{s}
  \BIaneg\phi$, so $r \satisfies \BICt{a}{s} \phi$ \mbox{\rm iff} $r
  \satisfies \BIaneg \BIC{a}{s} \BIaneg\phi$  iff $r \not\satisfies \BIC{a}{s}
          \BIaneg\phi$ iff there exists $r' \in R$ s.t. $r
          \bullet s \sim_a r'$ and $r' \not\satisfies \BIaneg\phi$ 
	iff there exists $r'\in R$ s.t. $r \bullet s \sim_a r'$ and
        $r' \satisfies \phi$.   
Proof of 2 is similar. 
\end{proof}

More intuitively,  we can see that $\BICt{a}{s}\phi$ expresses that the
agent, $a$, can establish the truth of $\phi$ if there exists a
resource such that the combination of the ambient resource, $r$, and
the local resource, $s$, is judged by $a$ to be equivalent to that
resource. Similarly, $\BIDt{a}{s}\phi$ expresses that the agent, $a$,
can establish the truth of $\phi$ using a resource that is the
combination of its local resource, $s$, with any resource such that
$a$ judges the combined resource to be equivalent to the ambient
resource, $r$.  We shall see later that these dual modalities can be also
useful for modelling systems.   

Returning to the possible representation of the modalities in an
hybrid version of \ESL, we could then define these modalities as
follows:  $\BICt{a}{s}\phi \equiv (\top \BImast n_s) \BIaand \BIKt{a}\phi$  
and  $\BIDt{a}{s}\phi \equiv \BIK{a} ((\top \BImast n_s) \BIaimp \phi)$, 
with  $n_s$ being a nominal forcing the resource $s$. As we have previously 
explained, here we aim at avoiding confusion between resources (which are 
part of the model) and propositions (which are part of the language) that 
we consider to be inconvenient for our intended modelling applications.

Note that the first point of the definition of $\bullet$, in
Definition~\ref{def_prm}, implies that the three other definitions
(neutral element, commutativity, and associativity) extend to
$\cdot$, so that the following are semantically equivalent  (i.e.,
every valid formula in the one is valid in the other) for any agent
$a$ and any  resources $r$, $s$, and $t$:  $\BIC{a}{re} \phi$ $\equiv$
$\BIC{a}{r} \phi$, $\BIC{a}{rs}$ $\equiv$ $\BIC{a}{sr}$, and
$\BIC{a}{r(st)}$  $\equiv$ $\BIC{a}{(rs)t}$. Of course, such 
equivalences also hold for $\BID{}{}\phi$, $\BICt{}{}\phi$, and
$\BIDt{}{}\phi$.

\section{Some properties of \ERL} \label{sec:properties}

We show that  \ERL\ is a conservative extension of Boolean \BI\ (\BBI)
and Epistemic Logic (\EL) and that, in the presence of additional
properties of the partial resource monoid (Definition \ref{def_prm}),
there are some noteworthy relationships between modalities.

We consider two fragments of \ERL. First, ${\rm \ERL}_{\rm BBI}$ ---
corresponding to BBI \cite{OHea01a} --- with $A=\emptyset$ on the
language $\mathcal{L}_{\mid BBI}$  defined as $\lang$ excluding the
$\BIC{a}{s}$ and $\BID{a}{s}$ operators. Second, ${\rm \ERL}_{\rm EL}$
--- corresponding to the epistemic logic EL consisting of classical
propositional additives and the basic epistemic operator $\EK{a}$
\cite{EpiHandbook} ---  with $Res=\{e\}$, on the language
$\mathcal{L}_{\mid EL}$ defined as $\lang$ excluding  $\BImI$,
$\BImast$, and $\BImimp$ and with $\BIC{a}{s}$ and $\BID{a}{s}$,
replaced  by the operator $\EK{a}$, which is defined, for all agents
$a$, by $\EK{a}\phi = \BIC{a}{e} \phi = \BID{a}{e} \phi$.  

\begin{proposition}[\ERL\ is a conservative extension of BBI and
  EL] \label{th_ral_extension} If, in every model of BBI, the neutral
  element of the composition is the element $e$ of $Res$,  then {\rm
    \ERL}$_{\rm BBI}$ is semantically equivalent to Boolean BI
  (BBI). If the agent sets  are the same for the two languages, {\rm
    \ERL}$_{\rm EL}$ is semantically equivalent to the epistemic logic EL.  
\end{proposition} 

We now consider some properties of \ERL; specifically, the way in
which the different operators  behave when they are used together in
formulae. One interesting property we might require in  our semantics,
which is based  on monoidal structure, is the compatibility of
$\sim_a$ and  $\bullet$. More precisely,  we might require that if two
resources are equivalent for an agent  $a$, then the  composition with
a third resource be transferred through this equivalence.    

Although such a property can be very useful, it introduces, from the
modelling perspective,  some quite strong properties: the transmission
of properties of resources through agent-dependent equivalence is a
strong assertion regarding  agents' private accesses, and should be
avoided when modelling some security  properties.  

Considering these concerns, we take this extra property to be
optional, and identify it in a sublogic of \ERL\  which we call
\ERLast.  

\begin{definition} \label{def_erl*}
The logic \ERLast\ is defined as \ERL\ with the addition of the
following property to the partial resource monoid (Definition
\ref{def_prm}): \\
For any agent $a$ and any resources $r,r' \in R$, if $r \bullet s 
\downarrow$ and  $r \sim_a r'$, then $r' \bullet s \downarrow$ and $r
\bullet s \sim_a r' \bullet s$. 
It is called the compatibility of $\sim_a$ with $\bullet$.
\end{definition}

Note that we use the logic \ERLast\ in the security modelling examples 
that we develop in the next section. 

\begin{lemma} \label{lem:eq} 
Let $a \in A$ be an agent, $s,t \in Res$ be resources and $\phi$ be a
formula of \ERLast.  
We have the following properties: \\  
\parbox{4.0cm}{\begin{itemize}
\item[\rm 1.] $\BIC{a}{s}(\BIC{a}{t}\phi)\equiv \BIC{a}{st}\phi$
\item[\rm 2.] $\BID{a}{s}(\BID{a}{t}\phi)\BIaimp \BID{a}{t}\phi$
\item[\rm 3.] $\BIC{a}{s}\phi  \BIaimp \BIDt{a}{t}(\BIC{a}{s}\phi)$
\end{itemize}} 
\parbox{4.0cm}{\begin{itemize}
\item[\rm 4.] $\BID{a}{t}(\BICt{a}{s}\phi) \BIaimp \BICt{a}{s}\phi$.
\item[\rm 5.] $\BICt{a}{t}(\BICt{a}{s}\phi)\equiv \BICt{a}{ts}\phi$
\end{itemize}} \quad 
\parbox{4cm}{\begin{itemize}
\item[\rm 6.] $\BIDt{a}{s}\phi \BIaimp \BIDt{a}{t}(\BIDt{a}{s}\phi)$
\item[\rm 7.] $\BIC{a}{e}\phi\equiv\BIDt{a}{e}\phi$
\end{itemize}}
\end{lemma}
\begin{proof}
First consider 1. Let $\mathcal{W}$ be a model and $r$ be a
resource. Suppose that $r\satisfies\BIC{a}{s}(\BIC{a}{t}\phi)$. Then
we have $r\bullet s\downarrow$ and, for any $r'\in R$ such that $r
\bullet s \sim_a r'$, we have $r'\satisfies\BIC{a}{t}\phi$. Thus $r
\bullet s\downarrow$ and, for any $r'\in R$ such that $r \bullet
s\sim_a r'$, $r' \bullet t\downarrow$, and for any $r''\in R$ such
that $r'\bullet t \sim_a r''$, we have $r'' \satisfies\phi$. Consider
$r''' \in R$ such that $r\bullet s\bullet t \sim_a r'''$. By
reflexivity, we obtain $r \bullet s\sim_a r \bullet s$. Then with $r'
= r \bullet s$ and $r'' = r'''$, we have $r \bullet s \bullet t
\downarrow$ and $r''' \satisfies \phi$. Thus $r \satisfies
\BIC{a}{st}\phi$, and we can deduce that $\BIC{a}{s}(\BIC{a}{t}\phi)
\BIaimp \BIC{a}{st}\phi$.  

Now suppose that $r\satisfies\BIC{a}{st}\phi$. Then $r\bullet
s\bullet t \downarrow$ and, for any $r'''$ such that $r\bullet s\bullet
t \sim_a r'''$, we have $r''' \satisfies \phi$. As $r \bullet s \bullet t
\downarrow$, we have $r\bullet s \downarrow$. Let $r'\in R$ be such
that $r\bullet s \sim_a r'$. Then, by compatibility, $r'\bullet
t\downarrow$ and $r\bullet s \bullet t\sim_a r'\bullet t$. Let $r''$
be such that $r'\bullet t\sim_a r''$. Then, by transitivity, we have
$r\bullet s \bullet t\sim_a r''$. Then, with $r'''=r''$, we have
$r''\satisfies\phi$. We obtain $r\bullet s\downarrow$ and, for any
$r'\in R$ such that $r\bullet s\sim_a r'$,  $r'\bullet t\downarrow$
and for any $r''\in R$ such that $r'\bullet t\sim_a r''$, we have 
$r'' \satisfies\phi$. Then we have
$r \satisfies\BIC{a}{s}(\BIC{a}{t}\phi)$, and then we can deduce 
$\BIC{a}{st}\phi \BIaimp \BIC{a}{s}(\BIC{a}{t}\phi)$. 
Finally, we have $\BIC{a}{s}(\BIC{a}{t}\phi)\equiv \BIC{a}{st}\phi$. 

Now consider 6. Let $\mathcal{W}$ be a model and $r$ be a resource. 
Suppose that $r\satisfies\BIDt{a}{s}\phi$. Then, for any $r'$ such
that $r' \bullet s\downarrow$ and $r\sim_a r'\bullet s$, we have
$r' \bullet s\satisfies\phi$. Let $r''$ such that $r'' \bullet t
\downarrow$ and $r \sim_a r''\bullet t$ and $r'''$ such that $r'''
\bullet s \downarrow$ and  $r'' \bullet t\sim_a r''' \bullet s$. By
transitivity we deduce that $r\sim_a r''' \bullet s$  and if we fix
$r'=r'''$ we have $r''' \bullet s\satisfies \phi$. As it is true for
any $r'''$ such that $r''' \bullet s \downarrow$ and $r''\bullet
t\sim_a r''' \bullet s$, we have $r'' \satisfies\BIDt{a}{s}\phi$. As
it is true that, for any $r''$ such that $r'' \bullet t \downarrow$
and $r \sim_a r'' \bullet t$, we have $r
\satisfies\BIDt{a}{t}(\BIDt{a}{s}\phi)$, then for any resource $r$ in
any model $\mathcal{W}$, $\BIDt{a}{s}\phi \BIaimp
\BIDt{a}{t}(\BIDt{a}{s}\phi)$ is valid.  

\medskip

Note that the reverse implication, $\BIDt{a}{t}(\BIDt{a}{s}\phi)
\BIaimp \BIDt{a}{s}\phi$, is not valid. In fact, if $r \satisfies
\BIDt{a}{t}(\BIDt{a}{s}\phi)$, $\phi$ is validated by all $r'' \bullet
s$ such that $r \sim_a r' \bullet t$ and $r' \bullet t\sim _a r''
\bullet s$. But to have $r \satisfies \BIDt{a}{s}\phi$, we must have
$r''' \bullet s \satisfies\phi$ for all $r'''$ such that $r\sim_a r'''
\bullet s$, and not only for those for which the equivalence by
$\sim_a$ is built from $t$. Then there is no equivalence between
$\BIDt{a}{s}\phi$ and $\BIDt{a}{t}(\BIDt{a}{s}\phi)$. 

\medskip 

All of the other cases are proved in similar ways. 
\end{proof}

We can complete our language with another modality  $\BIE{a}{s}\phi$
that could be also helpful for our modelling perspectives. From this
modality, that is a variant of $\BIC{a}{s} \phi$, we can also derive
$\BIEt{a}{s} \phi$ such that $\BIEt{a}{s} \phi \equiv  \BIaneg
\BIE{a}{s} \BIaneg\phi$. 
 
$\BIE{a}{s}\phi$ expresses that the agent, $a$, can establish
 the truth of $\phi$ using any resource combined with its local
 resource, $s$, provided $a$ judges that combination  to be
 equivalent to the combination of the local resource , $s$, with the
 ambient resource, $r$. 
 In other words $\BIE{a}{s}\phi$ is true relative to the ambient
resource $r$ iff for $a$'s views of the combination of the ambient
resource $r$ and its local resource $s$, $\phi$ is true. More formally
we have:   
\[
\begin{array}{rcl}

r \satisfies \BIE{a}{s} \phi & \mbox{\rm iff} & \mbox{\rm if $r \bullet s
\downarrow$ then for all $r' \in R$ s.t. $r' \bullet s \downarrow$ if
   $r \bullet s \sim_a r' \bullet s$, then $r' \bullet s \satisfies
       \phi$.}
\end{array}
\]

We can built $\BIE{a}{s} \phi$ from the previous main modalities as
follows.

\begin{proposition}
We have $\BIE{a}{s} \phi \equiv \BIC{a}{s}(\BIDt{a}{s} \phi)$.
\end{proposition}
\begin{proof}
Consider that  $r \satisfies\BIC{a}{s}(\BIDt{a}{s} \phi)$ iff, for
all $r' \in R$, if $r \bullet s \sim_a r'$, then $r' \satisfies
\BIDt{a}{s} \phi$ iff, for all $r' \in R$, if $r \bullet s \sim_a r'$,
then, for all $r'' \in R$, if $r' \sim_a r'' \bullet s$, then $r''
\bullet s \satisfies \phi$ iff, for all $r',r''\in R$, if $r  \bullet
s \sim_a r'$  and $r' \sim_a r'' \bullet s$, then  $r'' \bullet s
\satisfies \phi$  iff (by the transitivity of $\sim_a$), for all $r''
\in R$, if  $r \bullet s \sim_a r'' \bullet s$, then $r'' \bullet s
\satisfies \phi$ iff  $r \satisfies \BIE{a}{s} \phi$. 
\end{proof}

\section{Modelling access control with the logic \ERLast} 
\label{sec:modelling} 

In this section, we illustrate how to use \ERL, and its special
sublogic \ERLast, in modelling access control situations.  

Security policies, such as those for access control, are often 
formulated separately from the architectural context in which 
they are intended to be applied. This can lead to the existence 
of vulnerabilities. Specifically, when a particular security policy 
is applied to a particular system, the security properties of the 
resulting system may not be as intended. 

We aim to illustrate that the new operators $\BIC{a}{s}$ and
$\BID{a}{s}$ are appropriate for modelling situations where the access
to resources (whether they are locations or pieces of data) is
central. Indeed, both operators can be used to specify (in a 
slight different flavour) whether a resource verifies a property in
agent's $a$ perspective, granted that the local resource $s$ is
present. 

Before developing our examples, we recall that there exists a body of 
work based on Linear Logic (LL) and multiset rewriting for modelling
some access control problems in specific situations. For example,
multiset rewriting has been used to characterize security protocols
\cite{Cer01}. Our aim here, however, is to provide a more general
framework that can be a modelling tool in many situations rather than
be an ad hoc creation specific to a context. Even if such a framework
based on Linear Logic and modalities for authorization and knowledge 
exists \cite{Garg06}, we consider the differences between LL and BBI
that make the later a more convenient tool for modelling. Both are
able to model aspects of the properties of resources, but in LL
propositions represent resources while in \BBI\  (and, indeed, in BI)
propositions represent properties of resources that can be expressed
within the Kripke structures supporting  resource semantics. LL
focuses on the production and consumption --- essentially counting ---
of resources while \BBI\ focuses on separation and sharing of
properties on resources. Modal extensions of \BBI\ extend this view
to incorporate the production and consumption of resources via the
effects of actions in action modalities \cite{Cour15a,CGP15}.  

Because --- as explained in the introduction and in a substantial body 
of literature \cite{Pym02} --- the semantics of \BBI\ can  be
interpreted as being  a theory of resources and their properties, we
can directly use resources  as tokens in our modelling of systems
\cite{Pym09a}. Of particular note in this paper is the use of
\emph{local} resources. For example, $s$ in   $r\vDash\BIC{a}{s}$ is
of the same nature, but doesn't have the same role,  as the ambient
resource $r$. This allows a simple integration of new actors of a
system into a modelling using \ERL\ and avoids the creation  of new
formal elements of a more ad hoc nature.

\subsection{Modelling distributed systems} \label{subsec:dist-sys-mod} 

The construction of mathematical models always involves design choices. 
Our approach is guided the approach to modelling distributed systems
articulated in \cite{CMP12,AP16}. This approach builds upon the
observation that, from a slightly abstract yet convenient point of
view, the key structural components of a distributed systems are the
following: 
\begin{itemize}
\item[-] \emph{Locations}. The basic architecture of the system is
  considered to be described by a collection of connected
  places. Mathematically, we need some topological structure, with
  directed graphs be perhaps the most commonly useful set-up.  
\item[-] \emph{Resources}. Resources are situated at the locations
  identified in the system's architecture. They are the components of
  the system that are manipulated --- that is,  consumed, created,
  moved, and so on --- as the system evolves in order to the deliver
  the services that it is intended to provide. Mathematically, we
  take the `resource monoids' adopted in, for example, the 
  semantics of BI, in Separation Logic and, indeed,  in ERL. In the
  intuitionistic versions of these logics, we take a partially ordered
  (or sometimes preordered) partial monoid of resources. As we have
  seen in Section~\ref{sec:introduction}, the monoidal composition
  then captures the combination of resource elements and the ordering
  captures the comparison of resource elements. In the classical 
  versions, we drop the ordering and work just with combination.    
\item[-] \emph{Processes}. The services that a system provides are
  delivered by the execution of processes, during which resources are
  manipulated. Mathematically, in formal generality, we can describe
  processes using an algebraic calculus of processes. In
  \cite{Pym09a}, we have employed a variation  of Milner's basic
  system, SCCS \cite{Mil83}, adapted to capture the interaction with
  resources and locations.   
\end{itemize}
In addition, we require the following concept: 
\begin{itemize}
\item[-] \emph{Environment}. When a system is modelled, it is
  necessary to decide what is its boundary. Things that are outside of
  the boundary are not represented in detail within the
  model. Nevertheless, the model must interact with its
  environment. Mathematically, this can be represented stochastically, 
  using specified probability distributions to capture events at the
  boundary.    
\end{itemize}

The structural components collectively represent the state of a system
and can be used to define a process algebra with an operational
semantics that defines their co-evolution as actions occur
\cite{Pym09a,CMP12,AP16}:
\[
	L , R , E \stackrel{a}{\longrightarrow} L' , R' , E' \, . 	
\]

When building models in this style, it is necessary to set up a notion 
of \emph{signature} for a model. For basic actions $a$ and locations $L$, 
we define an evolution 
\[
	\mu(a , L , R) = (L' , R')
\]
that specifies the effect of $a$ on the resource $R$ at this $L$. We 
call $\mu$ a \emph{modification function}. 

In this setting, there is an associated modal logic with a satisfaction 
relation of the form 
\[
	L , R , E \models \phi , 
\]
which includes both additive and multiplicative action modalities
\cite{Pym09a,CMP12,AP16}. Additive action modalities yield 
formulae of the form $[a] \,\phi$, with a truth condition along the 
following lines: 
\[
\begin{array}{rcl} 
L , R , E \models [a] \, \phi & \;\mbox{iff}\; & \mbox{for all $E
                                                 \goes{a} E'$,\; $L' ,
                                                 R' , E' \models \phi$,} 
\end{array}
\]
where we need the condition, part of the signature of the model, to
the effect that the occurrence of the action $a$ causes the evolution
of $L$ to $L'$ and $R$ to  $R'$ \cite{Pym09a,CMP12,AP16}. The
multiplicative modalities allow actions to carry around local
resources that can be combined with the ambient resource  ---  so we
consider  $L , R , E \models [a]_S \, \phi$ and form $R' \circ S'$ in
the definiens of the satisfaction clause --- to enable the evolution
\cite{Pym09a,CMP12,AP16}.  

The logic is used both to constrain the model, through situation-specific 
logical properties, and to express desired or undesired properties of the 
system that are to be checked. 

In the setting of modelling access control using \ERL, locations,
resources, and  processes can all be represented, although we can make
some simplifications.   
\begin{itemize}
\item[-] \emph{Locations}. The examples we consider implicitly employ
 	location architectures, but they are sufficiently simple that they
  	can also be handled implicitly in the formalization, often
        through the treatment of resources. 
\item[-] \emph{Resources}. The resource elements considered carry the 
	structure of resource monoids, and we make essential use of this in
	the models.  
\item[-] \emph{Processes}. Our examples only deal with the actions
  that are required to instantiation the epistemic
  modalities. Nevertheless,	we provide discussions of how our
  examples can be understood in the location--resource--process context. 
\end{itemize}

In this setting, we elide the modelling of environment: since we are
not seeking to build executable models, this simplification is of
little or no consequence for our present purposes. In these senses, we
are making use of a fairly pure version of resource semantics. 

We employ a range of examples of security modelling using this
approach. We begin, in Section~\ref{subsec:schneier}, with `Schneier's
Gate', which illustrates the policy-architecture gap, and then
consider a core systems-security situations of joint access control,
in Section~\ref{subsec:joint}, and semaphores, in
Section~\ref{subsec:semaphores}.

\subsection{The `Schneier's Gate' problem} \label{subsec:schneier}

Consider the example of `Schneier's gate' \cite{Schneier2005}, wherein
a security system is ineffective because of the existence of a
side-channel that allows a control  to be circumvented. Here a
facility that is intended to be secured is protected by a  barrier
that prevents cars from entering  into the facility. The barrier may
be  controlled by a token --- such as a card, a remote, or a code ---
the holding of which  distinguishes authorized personnel from
intruders. If, however, the barrier itself is  surrounded by ground
that can be  traversed by a vehicle, without any kind of fence or
wall, then any car can drive around it (whether it's with a malicious
intent or just by  laziness of getting through  the security
procedure) and the access control policy, as  implemented by the
barrier and the tokens, is undermined. So, the access control  policy
--- that only authorized personnel, in possession  of a token, may
take vehicles into the facility --- is undermined by the architecture
of the system to which it is applied.    

\begin{figure}
\begin{center}
\includegraphics[scale=0.3]{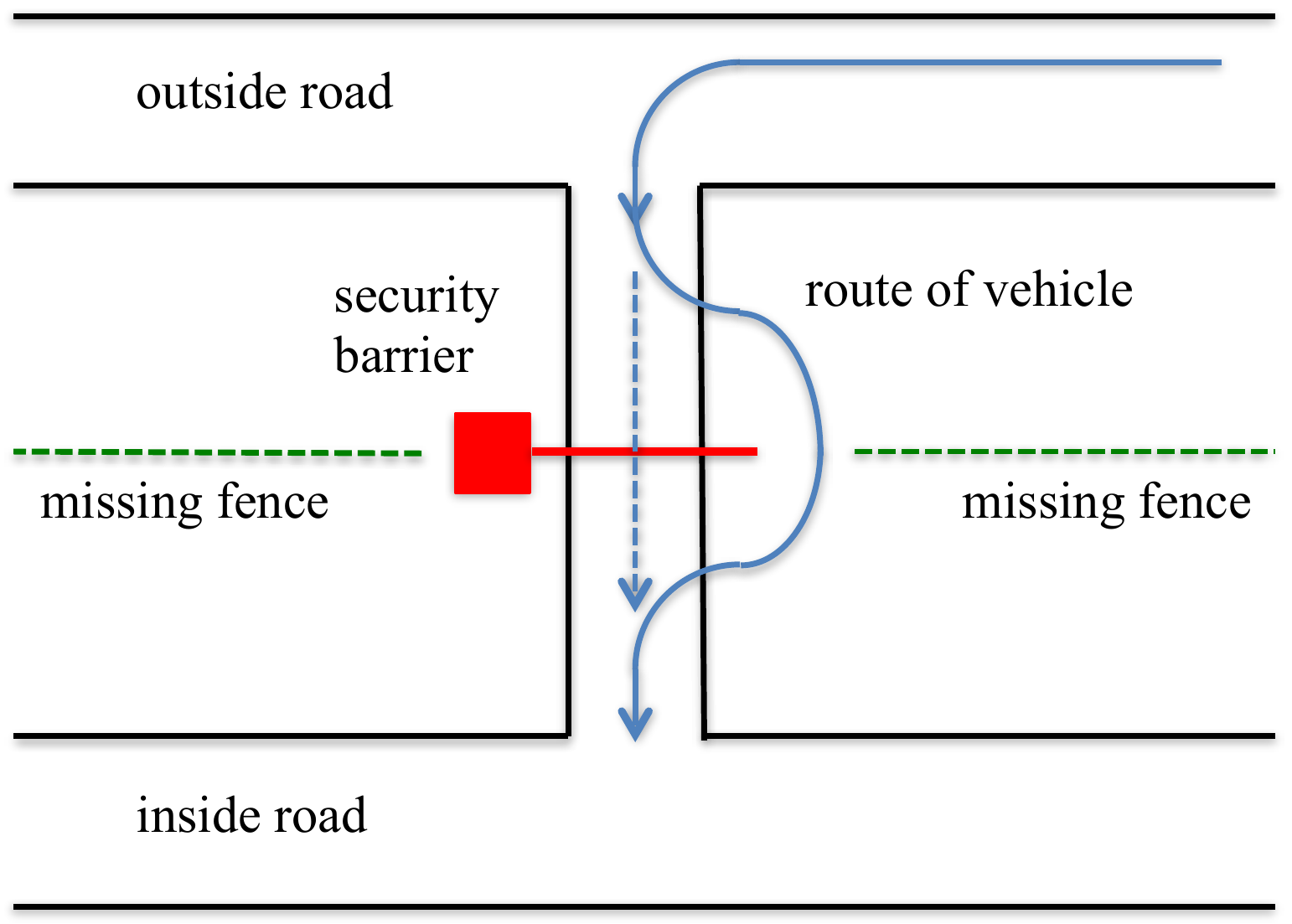}
\end{center}
\caption{A depiction of the `Schneier's gate' problem}
\end{figure} 


We show how \ERLast\ can be used to model, and so reason about, the 
situation described above (following \cite{Schneier2005}), illustrating 
how such situations can be identified by logical analysis. Related
analyses, employing logical models of layered graphs, can be found in 
\cite{CMP14}. 

We follow the approach to distributed systems modelling sketched in 
Section~\ref{subsec:dist-sys-mod} and elaborated in \cite{Pym09a,CMP12,AP16}. 
We start with a simple model, depicted in Figure
\ref{figure_barrier_1}, and gradually refine it. We model just a
facility protected by an access barrier. We will need the following 
key components:
\begin{itemize}
\item \emph{Locations}. We assume, for what is an architecturally 
	simple model, just three locations: \emph{outside} and
        \emph{inside} of the 	area guarded by the barrier, and the
        barrier itself. In this simple setting, there is no need to
        incorporate an explicit representation of locations into our
        model's worlds. 
\item \emph{Resources}. There are just three types of resource:
  vehicles (cars), access tokens, which are required to operate the
  barrier, and a marker for the presence of the barrier. 
\item \emph{Processes}. In this simple setting, we do not need to
  employ the full, quite complex, structure of a process algebra;
  rather, the actions of a logic with action modalities --- in
  particular, the action modalities of \ERLast, with their epistemic
  semantics, will suffice. 
\end{itemize}
In fact, our treatment of resource in this epistemic-logic setting is
a little more subtle. From the modelling perspective, the resources we
have exposed here are diverse in nature:  there is is a material token
(key or card for instance), there are cars, and a just a marker for
the presence and well-functioning of the barrier. This diversity
raises the question of the meaning and value of the unit resource,
$e$.  We finesse this problem by accepting that resources encompass a
variety of different objects, but we can also employ the epistemic
nature of our logic and consider that resources represent not objects
as such but rather the  knowledge that a given object is in our system. 

A vehicle having the appropriate access token should be able to get
inside. We consider the following sets of resources, agents, and
logical properties of resources/system states: 
\[
Res = \{ e , b , t , c \}, \ A = \{ \alpha \}, \ \Prop = \{ O , J
\}. 
\]  
Here we have the following: 
\begin{itemize}
\item the atomic propositions $O$ and $J$, respectively, express the
  state of being \emph{outside} and \emph{inside} the facility --- we
  use $J$ instead of $I$ to avoid confusion with $\BImI$, the unit
  operator;  
\item a resource element $b$ is taken as a marker for the presence and
  well-functioning of the barrier;  
\item a token, required to operate the barrier, is denoted by a
  resource element $t$ and vehicles (cars) are denoted by resource
  elements $c$, $c'$, etc.;
\item for simplicity we are assuming that all resource elements are of
  the same sort; that is, are elements of the same resource monoid;
  this will cause no formal difficulty in this simple setting, though
  richer examples might require more care in this respect;  
\item $u \satisfies O$ means that $u$ is outside the facility, and $v
  \satisfies J$ means that $v$ is inside.   
\item the agent $\alpha$ is a generic one that represents a user of
  the system; that is, say, the vehicle/driver that approaches the
  access control point. The resources $b$ and $t$ represent tokens
  that stand respectively for the barrier and the access token of the
  users. 
\end{itemize}
So, $c$ can be viewed as an abstract token marking the presence of a
car, and $t$ the presence of the required access device in this
car. Thus resources act as an abstraction layer of our system. In this
view, it follows that it is easy to see $e$ as the absence of
information (nothing is known of the system).  


\begin{figure}[t]
\hrule
\vspace{1mm}
  \centering
  \begin{minipage}[t]{0.4\textwidth}
  \begin{tikzpicture}[scale=0.6, >=latex]

  	\tikzstyle{carre}=[draw,rectangle]
  	\tikzstyle{arrete}=[-,thick]
  	\tikzstyle{barrier}=[-,line width=5pt]
  	\tikzstyle{car}=[-,dotted,thick, red]
  
 	\node at (-1,-1)(margin1) {};
  	\node at (6,5)(margin2) {};

	\draw[arrete] (0,0) -- (5,0);
 	\draw[arrete] (0,4) -- (5,4);
  
	\draw[arrete] (0,1) -- (2,1);
	\draw[arrete] (2,1) -- (2,3);
  	\draw[arrete] (2,3) -- (0,3);
  
	\draw[arrete] (3,1) -- (5,1);
 	\draw[arrete] (3,1) -- (3,3);
	\draw[arrete] (3,3) -- (5,3);


	\draw[barrier] (1.5,2) -- (3.5,2);
	\node at (3.7,2.2)(barrier_token) {$b$}; 
  

	\draw[car] (5,3.5) -- (2.5,3.5);
	\draw[car] (2.5,3.5) -- (2.5,1);
	\node at (2.1,3.7)(car_token) {\color{red}$c,t$}; 


	\node at (1,3.5)(outside) {$O$}; 
	\node at (2.5,0.5)(inside) {$J$};

  \end{tikzpicture} \vspace{-10mm}
  \caption{Barrier problem, base case} 
  \label{figure_barrier_1} 
  \end{minipage} \qquad\qquad
  \begin{minipage}[t]{0.4\textwidth}
  \begin{tikzpicture}[scale=0.6, >=latex]

  \tikzstyle{carre}=[draw,rectangle]
  \tikzstyle{arrete}=[-,thick]
  \tikzstyle{barrier}=[-,line width=5pt]
  \tikzstyle{car}=[-,dotted,thick, red]

  \node at (-1,-1)(margin1) {};
  \node at (6,5)(margin2) {};

  \draw[arrete] (0,0) -- (5,0);
  \draw[arrete] (0,4) -- (5,4);
  
  \draw[arrete] (0,1) -- (2,1);
  \draw[arrete] (2,1) -- (2,3);
  \draw[arrete] (2,3) -- (0,3);
  
  \draw[arrete] (3,1) -- (5,1);
  \draw[arrete] (3,1) -- (3,3);
  \draw[arrete] (3,3) -- (5,3);


\draw[barrier] (1.5,2) -- (3.5,2);
 \node at (3.7,2.2)(barrier_token) {$b$}; 
  

\draw[car] (5,3.5) -- (2.5,3.5);
\draw[car] (2.5,3.5) -- (2.5,1);
\node at (2.0,3.7)(car_token) {\color{red}$c,t_{\alpha}$}; 


\node at (1,3.5)(outside) {$O$}; 
\node at (2.5,0.5)(inside) {$J$};

  \end{tikzpicture} \vspace{-10mm}
  \caption{Barrier problem with agents} 
  \label{figure_barrier_2}
  \end{minipage} 
\vspace{1mm}
\hrule 
\end{figure}

We have the following property: $O \BIaimp \BIC{\alpha}{bt} J$.  
According to the semantics, based on a resource monoid $R$, $c
\satisfies O \BIaimp \BIC{\alpha}{bt} J$ just in case if $c\satisfies
O$, then, for every $c' \in R$ such that $c \bullet b \bullet t 
\sim_{\alpha} c'$, $c'\satisfies J$.  
Thus the combination of the two tokens grants access to the
inside. The use of the token $b$ for the presence of the barrier helps
in modelling a situation in which the barrier is completely shut or is
broken (in which case entering wouldn't be possible).  Note that the
formulae $O \BIaimp \BIC{\alpha}{t} J$, $O \BIaimp \BIC{\alpha}{b} J$,
and $O \BIaimp \BIC{\alpha}{e} J$ are not valid because we cannot
enter if the barrier is shut, if we have no access token, or both. 

The use of the operator $\BIC{\alpha}{s}$ in this situation is
illustrative. First, consider what differences the use of other
operators  would make. If we were to state $O \BIaimp
\BIDt{\alpha}{bt} J$, then it  would mean that anyone outside can get
(without condition) inside  and acquire the two access tokens. This is
of course not what we expect.   
On the other hand, using $\BIE{\alpha}{s}$ has an interesting effect. $O
\BIaimp \BIE{\alpha}{bt} J$ requires not only that an entering agent have  
the expected tokens, but also that those tokens remain active once
they are inside. This is slightly different from our first approach: 
we don't know if the tokens are still active once the agent is
inside.  

We can also consider which of the additive implication, $\BIaimp$, and  
the multiplicative, $\BImimp$, would be the better modelling choice in 
this example. For a first approach, $\BIaimp$ seems quite sufficient. 
Indeed, if we assert $O \BIaimp \BIC{\alpha}{bt} J$ as valid, then any
resource satisfies it. So, if we have a car $c$ such that $c\satisfies
O$, we also have $c\satisfies O \BIaimp \BIC{\alpha}{bt} J$, and then
we get the expected $c\satisfies\BIC{\alpha}{bt} J$. 

However, if we consider more complex properties, the situation is
different. Imagine, for example, an environment that is composed not
only of the car $c$,  but also another entity, or piece of
information, $o$. Our epistemic context is thus $o \bullet c$. If we
have $c\satisfies O$ and if $O \BIaimp \BIC{\alpha}{bt} J$  is valid,
then we get $c\satisfies \BIC{\alpha}{bt} J$. As we do not have  $o
\bullet c \satisfies O$, we cannot deduce that  $o \bullet c
\satisfies \BIC{\alpha}{bt} J$.  

If instead we assume that the property  $O \BImimp \BIC{\alpha}{bt} J$
is valid, then we have, in particular, $o\satisfies O \BImimp 
\BIC{\alpha}{bt} J$ and, together with $c\satisfies O$, we can deduce
$o\bullet c\satisfies\BIC{\alpha}{bt} J$, as desired. So,  the use
of $\BImimp$ instead  of $\BIaimp$ is much more useful in more 
complex systems, as it allows us to set aside, as with Separation
Logic's Frame Rule, some of the entities of our system and still
apply the property. 

Now we introduce agents to the model (see Figure
\ref{figure_barrier_2}). The first model may seem crude, because a
single resource is used to model the access of any agent.  So, we seek
to benefit from the logic that allows us to take agents into account.  

We change the model by defining a detailed set of agents, 
$A = \{ \alpha , \beta , \gamma \}$ and now take three agents or
users, $\alpha$, $\beta$, and $\gamma$. Each user should have its own
access token, and the resource set is modified accordingly:  $Res = \{
e , b , t_\alpha , t_\beta , t_\gamma , c \}$.  Now the slightly different
formula $O \BIaimp \BIC{a}{bt_{a}} J$ is  valid for any agent $a \in
A$. So, for example,  $O \BIaimp \BIC{\alpha}{bt_\alpha} J$  is 
valid, which means that $\alpha$ can  get inside with his own token,
but $O \BIaimp \BIC{\alpha}{bt_\beta} J$ is not, which means $\alpha$
cannot use $\beta$'s token.  

\begin{figure}[t]
\hrule
\vspace{1mm}
  \centering
  \begin{minipage}[t]{0.4\textwidth}
  \begin{tikzpicture}[scale=0.6, >=latex]

  \tikzstyle{carre}=[draw,rectangle]
  \tikzstyle{arrete}=[-,thick]
  \tikzstyle{barrier}=[-,line width=5pt]
  \tikzstyle{car}=[-,dotted,thick, red]
  \tikzstyle{car2}=[-,dotted,thick, blue]

  \node at (-1,-1)(margin1) {};
  \node at (6,5)(margin2) {};

  \draw[arrete] (0,0) -- (5,0);
  \draw[arrete] (0,4) -- (5,4);
  
  \draw[arrete] (0,1) -- (2,1);
  \draw[arrete] (2,1) -- (2,3);
  \draw[arrete] (2,3) -- (0,3);
  
  \draw[arrete] (3,1) -- (5,1);
  \draw[arrete] (3,1) -- (3,3);
  \draw[arrete] (3,3) -- (5,3);


\draw[barrier] (1.5,2) -- (3.5,2);
 \node at (3.7,2.2)(barrier_token) {$b$}; 
  

\draw[car] (5,3.5) -- (2.5,3.5);
\draw[car] (2.5,3.5) -- (2.5,1);
\node at (1.9,3.7)(car_token) {\color{red}$c,t_\alpha$}; 


\draw[car2] (5,3.6) -- (2.4,3.6);
\draw[car2] (2.4,3.6) -- (2.4,2.5);
\draw[car2] (2.4,2.5) -- (1,2.5);
\draw[car2] (1,2.5) -- (1,1.5);
\draw[car2] (1,1.5) -- (2.4,1.5);
\draw[car2] (2.4,1.5) -- (2.4,1);
\node at (0.7,2)(car_token) {\color{blue}$\beta$}; 


\node at (1,3.5)(outside) {$O$}; 
\node at (2.5,0.5)(inside) {$J$};
  \end{tikzpicture} \vspace{-10mm}
  \caption{Barrier problem with a shortcut} 
  \label{figure_barrier_3}
  \end{minipage} \qquad\qquad 
  \begin{minipage}[t]{0.4\textwidth}
  \begin{tikzpicture}[scale=0.6, >=latex]

  \tikzstyle{carre}=[draw,rectangle]
  \tikzstyle{arrete}=[-,thick]
  \tikzstyle{barrier}=[-,line width=5pt]
 \tikzstyle{fence}=[-,dotted, line width=5pt]
  \tikzstyle{car}=[-,dotted,thick, red]
  \tikzstyle{car2}=[-,dotted,thick, blue]
  
  \node at (-1,-1)(margin1) {};
  \node at (6,5)(margin2) {};

  \draw[arrete] (0,0) -- (5,0);
  \draw[arrete] (0,4) -- (5,4);
  
  \draw[arrete] (0,1) -- (2,1);
  \draw[arrete] (2,1) -- (2,3);
  \draw[arrete] (2,3) -- (0,3);
  
  \draw[arrete] (3,1) -- (5,1);
  \draw[arrete] (3,1) -- (3,3);
  \draw[arrete] (3,3) -- (5,3);


\draw[barrier] (1.5,2) -- (3.5,2);
 \node at (3.8,2.5)(barrier_token) {$b$}; 
  

\draw[car] (5,3.5) -- (2.5,3.5);
\draw[car] (2.5,3.5) -- (2.5,1);
\node at (1.9,3.7)(car_token) {\color{red}$c,t_\alpha$}; 


\draw[car2] (5,3.6) -- (2.4,3.6);
\draw[car2] (2.4,3.6) -- (2.4,2.5);
\draw[car2] (2.4,2.5) -- (1,2.5);
\draw[car2] (1,2.5) -- (1,2);


\node at (1,3.5)(outside) {$O$}; 
\node at (2.5,0.5)(inside) {$J$};

\draw[fence] (0,2) -- (5,2);
 \node at (0,2.5)(fence) {$F$}; 
  \end{tikzpicture} \vspace{-10mm}
  \caption{Barrier problem with a fence} 
  \label{figure_barrier_4}
  \end{minipage}
 \vspace{1mm}
 \hrule
\end{figure}

Now consider the case in which the access is controlled and the agents
are supposed to cross the barrier only if they have the appropriate
access device. We want to capture the  fact that the system can
actually be flawed (as mentioned in the problem presentation).  It is
actually quite easy to do, because being able to circumvent the
barrier just means  being able to access inside of the complex without
any token. We could be a little  more specific by imagining that some
agents know the shortcut (or dare to use it) and  others don't (See
Figure \ref{figure_barrier_3}). In the previous setting, suppose that
the agent $\beta$ is aware of the shortcut and is disposed to use
it. Our new set of  properties should now be the following:    

\[
\left \{
\begin{array}{l}
    O \BIaimp \BIC{a}{bt_{a}} J \;(\mbox{for every $a \in A$}), \; 
    O \BIaimp \BIC{\beta}{e} J  
\end{array}
\right \} . 
\]

The unit resource $e$ expresses a direct access (with no resource
needed). Note how the use of agents can help us to express different
security policies in the same model. 
 
We can reasonably suppose that such a flawed system would be quickly
dealt with; for example, by installing a fence that would prevent
going around the barrier (See Figure  \ref{figure_barrier_4}). We
could, of course, just model that by removing our last addition  and
get back to the intended policy, but it is more  interesting to encode
it by a formula.  For example, we might then also  describe a fault in
the fence (or its removal). To do so, we  can simply add a
propositional formula $F$ that is valid for any resource provided
there  is a fence  preventing the passage of `rogue' agents. Our
system then becomes  
\[
\left \{
\begin{array}{l}
    O \BIaimp \BIC{a}{bt_{a}} J  \;\mbox{(for every $a \in A$)}, \; 
    O \BIaand \BIaneg \ F \ \BIaimp \BIC{\beta}{e} J 
\end{array}
\right \} .
\]

Having established a system of formulae that describes our modelling
situation quite clearly, we can seek to some properties of the
model. The idea is to establish a property of the system  that goes
beyond its basic definition. For example, we may want to check that
every agent  inside the facility has passed the barrier and has in its
possession its access token. This  means that we must prove that, for
every agent $a \in A$, $J \BIaimp \BID{a}{bt_{a}} J$.  

Indeed, if $c \satisfies J \BIaimp \BID{a}{bt_{a}} J$, this means that 
if $c\satisfies J$, then there exists $c' \in R$ such that $ c \sim_{a} 
c'\bullet b\bullet t_{a}$ and  $c'\bullet b\bullet t_{a} \satisfies J$, 
which expresses that every resource representing  a car that is inside
must in fact be equivalent, for an agent $a \in A$, to a resource that
is inside \emph{and} is composed with both the appropriate token
$t_{a}$  and the barrier token $b$. This is exactly what we wanted to
capture. 
 
Notice that this  particular property is not verified by the system we
described in our set up. Indeed,  noted previously, specifying
entrance with $r\satisfies O\BIaimp \BIC{a}{bt_{a}}J$ makes $J$ be
satisfied by any resource $r'$ such that $r\bullet b\bullet
t_{a}\sim_{a} r'$. We can see that $r'$ does not contain $b$ and
$t_{a}$. The use of $\BIE{a}{bt_{a}}$ instead  solves this problem: we
then have $r\bullet b\bullet t_{a} \sim_{a} r'\bullet b\bullet 
t_{a}$  and $r'\bullet b \bullet t_{a} \satisfies J$, as required.

So far, we have considered only simple situations, mainly one car crossing 
the barrier in various situations. Of course, we may wish to consider more 
complex models and establish similar properties. For example, we may 
want to see what happen if several cars are modelled together in the system. 

We have the sets of properties in the form of implications stated before. 
To state there is a car in the system, we just assert that the formula
$O$ is valid. Then, by looking at the semantics of our formulae, we
create a resource $c$ which satisfies that formula. In order to have
several cars, we might at first be tempted to assert something like 
$O \BIaand O \BIaand O$ (for three cars). However, given our semantics,  we have
trivially that $O \BIaand  O \BIaand O \equiv O$, which is inconvenient for
our modelling purpose. It is better to state $O\BImast  O\BImast O$, using
the multiplicative conjunction, instead.  Then, to satisfy this
formula, we need indeed three resources $c_1,c_2,c_3$  and we have
$c_1 \bullet c_2 \bullet c_3 \satisfies O \BImast O \BImast O$ --- that is, 
for each car to gain access, a a token is required for that car. 
Then, using $\BImimp$ as described  above, we can see the system
evolve as cars are allowed inside. Thus, the use of $\BImast$  is
particularly relevant to model several instances of a same object. 

Of course, we could easily enrich this model to make more distinctions 
between different cars and their different properties, but the essentials 
of the model would remain the same.

\subsection{Joint access} \label{subsec:joint}

One of the most common problems of access control is joint access and 
we propose to model a very simple example with our logic. The background 
for this example can be found in many films about the cold war era: the situation 
is that a critical system --- such as one that controls the release of
nuclear weapons, as in `Crimson Tide' \cite{CrimsonTide} --- is
secured by two different keys, each one  held by a different
operator. For the system to unlock, it is necessary that both
operators activate their keys simultaneously.  We provide a logical 
analysis of this situation. 

From our systems modelling perspective, we can set this up quite simply, 
as depicted in Figure~\ref{fig:joint}. 

\begin{figure}[h]
\hrule 
\vspace{2mm}
\begin{center}
\includegraphics[scale=0.4]{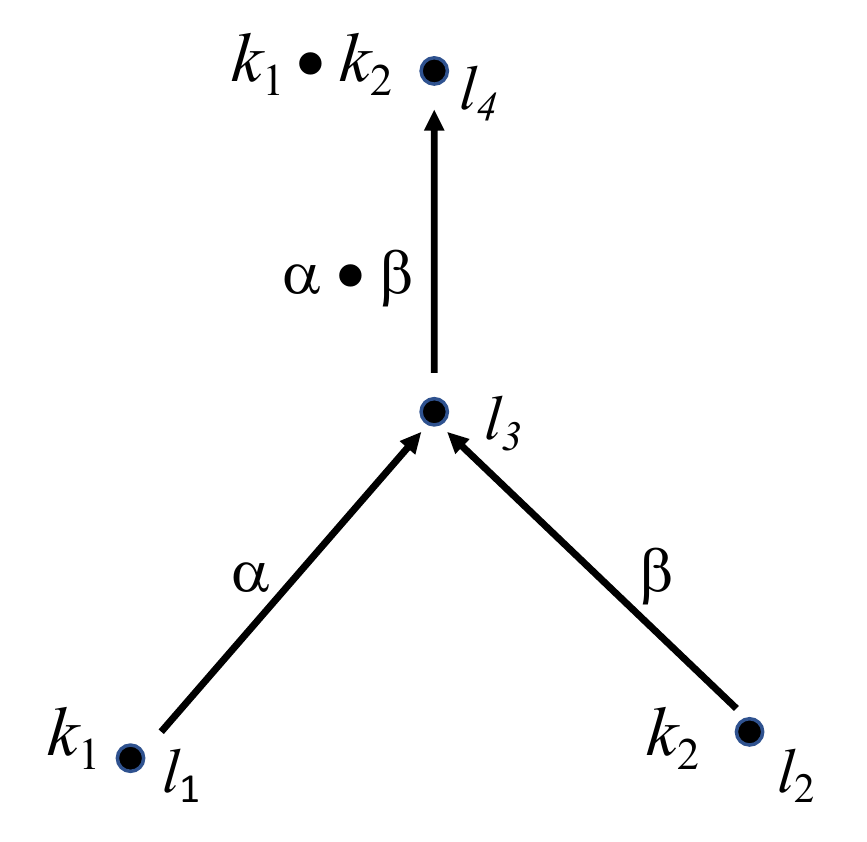}
\end{center}
\vspace{-7mm}
\caption{Joint access} \label{fig:joint}
\vspace{2mm}
\hrule
\end{figure} 

Some of the modelling choices made here are quite obvious: we need two agents,
and two associated resources representing their keys. So, we take
$A = \{\, \alpha , \beta \,\}$ and $Res = \{\, k_1 , k_2 , e
\,\}$. Implicitly, the formulae will express that $\alpha$ is
associated to $k_1$ and $\beta$ to $k_2$. Also  implicitly, we are
employing four locations, $l_1$ -- $l_4$, so that we can sketch a
system model as  
\[
\begin{array}{rcl}
l_1 \, , \, k_1 \, , \,  \alpha : Unlock_1 : 0 & \goes{\alpha} & l_3 \, , \,  k_1 \, , \, Unlock_1:0 \\ 
l_2  \, , \,  k_2  \, , \,  \beta : Unlock_2 : 0 & \goes{\beta} & l_3  \, , \,  k_2  \, , \,  Unlock_2:0     \\
l_3  \, , \,  k_1 \bullet k_2  \, , \,  \underbrace{Unlock_1 : 0 \times Unlock_2 : 0}_{\stackrel{\mbox{def}}{=} \; Unlock} 
	& \goes{\alpha \bullet \beta} & l_4  \, , \,  k_1 \bullet k_2  \, , \,  0 , 
\end{array}
\]
where $l_3 \bullet l_3 \stackrel{\mbox{def}}{=} l_3$, and where the modification function of the 
model, which describes how the keys move from location to location, is given by 
\begin{itemize}
\item[-] $\mu(\alpha, l_3 , k_1 \bullet k_2) = (l_4 , k_1\bullet k_2)$, 
\item[-] $\mu(\alpha , l_1 , k_1) = (l_3 , k_1)$, and 
\item[-] $\mu(\alpha , l_2 , k_2) = (l_3 , k_2)$. 
\end{itemize}

Focussing on our logical modelling, and suppressing for now the location 
architecture, we must express the fact that each agent --- representing here 
a simplified notion of process --- must use its key. Of course, as the whole 
point of the example is to illustrate how two separate accesses unlock the 
system, thus each use of key must be modelled with a different formula. 
We propose the following formulae for this purpose: 
\[
	\BID{\alpha}{k_1}\BIatop \ and \ \BID{\beta}{k_2} \BIatop . 
\]
We use the atomic formula $\BIatop$ since we don't need to access any
property --- rather we need only to update $\alpha$ and $\beta$'s
accessible worlds to express that $k_1$ and $k_2$ are now
activated. If we consider $\BID{\alpha}{k_1}\BIatop$ for instance,
then if $r \satisfies \BID{\alpha}{k_1}\BIatop$, then there exists a
resource  $r'$ such that $r\sim_\alpha r'\bullet k_1$ and $r'\bullet
k_1 \satisfies\BIatop$. Given this last statement, we have  that there
exists $r'$ such that $r\sim_\alpha r' \bullet k_1$. Thus, with this
formula we have stated that $\alpha$ can reach a state in which $k_1$
is activated. The second formula states the same for $b$ and $k_2$.   

We must express that whenever both keys are present, the system can
be unlocked. We could consider using a formula such as
$\BIDt{\alpha}{k_1 k_2}U$,  where $U$ is an atomic formula expressing
that the system is unlocked.  However, we can see at once that this
choice is problematic. Indeed, this formula  is dependent on $\alpha$,
but the point of joint access is that none of the  agents involved is
responsible on its own for the activation of the device.  Moreover,
should we decide to proceed with such a formula, it would fail  to do
the required job --- $k_2$ is brought in the system by $\beta$ and
only $\alpha$ is present in the formula. Obviously, using $\beta$
instead of $\alpha$ raises the  same problems (symmetrically). \\

It seems, therefore, that our model lacks (at least) an agent. We
introduce an omnipotent agent $o$ (and thus $A = \{ \alpha , \beta , o
\}$). The idea is to have an agent that can see and use whatever $\alpha$
and $\beta$ can, without the two sharing knowledge or potential
action. This agent can be interpreted either  as a global authority or
just as a modelling of the device itself (the computer that accepts
the keys and executes the order). Now, with this extra agent,
$\BIDt{o}{k_1 k_2}U$ seems to be an acceptable candidate for modelling  
the unlocking of the system. This states that whichever state reachable for 
$o$ that contains $k_1$ and $k_2$ triggers the unlocking. However, 
we still need to express $o$'s capability. To do that, we introduce the 
following set of formulae: 
\[
\left \{ \BID{a}{s} \phi \rightarrow \BID{o}{s} \phi \mid a \in A,\ s \in
  Res, \ \phi \in \lang \right \}. 
\]
This expresses that any access to a resource by an agent through
the modality $\BID{}{}$ can be transferred to $o$. Of course, in a
more general setting, we could state similar things  for the other
operators, but, in this very particular example, only $\BID{}{}$ will
be useful. \\

Finally, in order to the system to work, we need to activate both keys
simultaneously. A first  approach could be to append the two
key-activation with an $\BIaand$: 
$\BID{\alpha}{k_1}\BIatop \BIaand  \BID{\beta}{k_2}\BIatop$. This
doesn't produce the desired result. Indeed, if
$r\satisfies\BID{\alpha}{k_1}\BIatop \BIaand \BID{\beta}{k_2}\BIatop$,
then we get  $r\sim_\alpha r'\bullet k_1$ and $r\sim_\beta r''\bullet
k_2$ and we intended to have the combination of $k_1$ and $k_2$, which
is here not obvious. Thus, the best way is in fact to use
$\BID{\alpha}{k_1}\BIatop \BImast \BID{\beta}{k_2}\BIatop$.  
More than the simple correctness of our modelling, this use of
$\BImast$ is quite  convincing, as we aimed to model the separated use
of two keys. 

Thus we have modelled our situation as follows: 
\begin{enumerate}
\item $\forall \ ag \in A,\ \forall s \in Res,\ \forall \phi \in \lang,\
  \BID{ag}{s}\phi\BIaimp\BID{o}{s}\phi$;  
\item $\BID{\alpha}{k_1}\BIatop \BImast \BID{\beta}{k_2}\BIatop$; 
\item $\BIDt{o}{k_1 k_2}U$. 
\end{enumerate}

We can check that this has the desired effect; that is, that whenever 
both keys are present, the system can be unlocked. 
Consider a resource $r$ that forces (2) and (3). The forcing of (3),
unpacked, means  
\[
\mbox{\rm for all } r'\text{ such that } r\sim_o r'\bullet k_1\bullet k_2,\
r'\bullet k_1\bullet k_2\satisfies U . 
\]
On the other side, unpacking of (2) gives 
\[
\mbox{\rm there exist } r_1,r_2\text{ such that }r=r_1\bullet
r_2\text{ and }r_1\satisfies\BID{\alpha}{k_1} \BIatop \text{ and
}r_2\satisfies\BID{\beta}{k_2} \BIatop .   
\]
We can then instantiate (1) twice, with $ag=\alpha$, $s=k_1$, and
$\phi=\BIatop$, then with $ag=\beta$, $s=k_2$, and $\phi=\BIatop$ to get 
\[
\mbox{\rm there exist } r_1,r_2\text{ such that }r=r_1\bullet
r_2\text{ and }r_1\satisfies\BID{o}{k_1}\BIatop \text{ and
}r_2\satisfies\BID{o}{k_2} \BIatop . 
\]
Unpacking this, we get
\[
\mbox{\rm there exist } r_1,r_2,r_1', r_2'\text{ such that }r=r_1\bullet r_2\text{
  and }r_1\sim_o r_1'\bullet k_1 \text{ and }r_2\sim_o r_2'\bullet k_2 . 
\]
By the compatibility of $\bullet$ and $\sim$, we obtain that $r\sim_o
r'_1\bullet k_1\bullet r_2$  and then that $r\sim_o r'_1\bullet
k_1\bullet r'_2\bullet k_2$, which by commutativity is $r\sim_o
r'_1\bullet r'_2\bullet k_1 \bullet k_2$. Then we have  $r'_1\bullet
r'_2\bullet k_1 \bullet k_2 \satisfies U$, as required. 

\subsection{Semaphores} \label{subsec:semaphores}

Another important example of modelling in access control is concerned 
with concurrency in parallel programming. We have described in the
introduction how Separation Logic, built on BI, is a powerful and
efficient tool to model memory management. We propose, in this
section, an example of a similar work with ERL* in which we use it to
model programs accessing  memory and the particular example of simple
concurrency with semaphores. 

First, we establish the general basis of our modelling approach. We
consider a multi-processor (or a set of different systems) which is
seeking to run multiple programs or tasks with a limited amount of
memory space.   
\begin{itemize}
\item[-] The set $R$ of resources will represent the memory of the system, 
	$Res$ being a subset of the memory specified for each problem. $e$ always 
	denotes an empty set of information in the memory. Thus, in this example, 
	we again suppress location, conflating it with resource.  
\item[-] The set of agents $A$ represents all the different threads or processes 
	which are running the tasks. 
\item[-] Two parts, $m$ and $m'$, of the memory are linked by the relationship 
$\sim_\alpha$ if the access to $m$ is equivalent to the access to $m'$ for the process
	$\alpha$. 
\item[-] Finally, we  use propositions of ERL* to model programs run by
the thread. Thus, when we write $m \satisfies P$, we mean that the
memory stored in $m$ is used to run the program $P$.  
\end{itemize}
Just as in the example of joint access, we can set up our modelling of 
semaphores in the context of our general approach to systems modelling.
We suppress the details here, preferring to use the simplified approach 
afforded by the logical tools introduced in this paper, but see \cite{CGP15} for 
examples of similar models that more closely following the system
modelling approach.  

So, consider how to model semaphores in this context. Recall that 
semaphores are simple bits of program which use flags or tokens to
ensure that a specific portion of program, called \emph{critical section}, 
is always accessed by at most one process. We use an
arbitrary set of agents $A$, and the set of resources $Res = \{ e , t \}$, 
where $t$ is a token marking the entrance into the critical section. We
also have two propositions $C$ and $NC$, the former being the critical
section of code, the latter being all the non-critical part of the
code. Note that, here, the agents correspond to processes.   

We consider the following formulae, which constrain the model, for any 
arbitrary process $\alpha \in A$:
\begin{enumerate}
\item \mbox{\rm $Guard$: for any $\alpha', \alpha'' \in A$ s.t. $\alpha' \neq \alpha''$, 
	$\BICt{\alpha'}{t}\BIatop \BIaimp \BIaneg \BICt{\alpha''}{t}\BIatop$};  
\item $In:  NC \BIaimp\BIC{\alpha}{t}C$; 
\item $Out : C\BIaimp((\BIaneg\BID{\alpha}{t}\BIatop) \BIaand \BID{\alpha}{e}NC)$.  
\end{enumerate}

The $Guard$ formulae, true for any two different processes $\alpha'$ and
$\alpha''$, ensure  that two processes cannot enter a critical section
together. Indeed, if, for any Guard formula, we have  that $m
\satisfies Guard$, then, if there is $m'$ such that $m\bullet
t\sim_{\alpha'} m'$,  there is no $m''$ such that $m\bullet t\sim_{\alpha''}
m''$. That is, for any process $p'$  which has the token $t$ in
memory, no other process $p''$ can get the token.  

The $In$ formula specifes that the process $\alpha$ enters the critical
section. If we  have that $m\satisfies In$, then, if $m\satisfies NC$,
then, for any $m'$ such that  $m\bullet t \sim_\alpha m'$, we have that
$m' \satisfies C$. That is, if a process is running the  non-critical
section, the addition of the token $t$ gives it access to a memory
state sufficient to run the critical section.  

Symmetrically, the $Out$ formula expresses the exit of $p$ from a
critical section. If $m\satisfies Out$, then, if $m\satisfies C$, 
then, for all $m'$ such that  $m \sim_\alpha m'\bullet t$, $m' \bullet \not
\satisfies \BIatop$. That is, there is no  $m'$ such that $m\sim_\alpha
m' \bullet t$. This allows us to delete $t$ from the memory accessible
by $\alpha$. The second part of the formula, $\BID{\alpha}{e}NC$, states that
there is a state $m''$ such that $m \sim_\alpha m''$ and $m'' \satisfies
NC$; that is, $\alpha$ gets back into non-critical section. 

No memory state that satisfies $NC$ after $C$ has been executed,  can
have $t$ in it.  So, once this formula is taken into account, either
$p$ can continue to execute $C$  or go into $NC$ and release the token
$t$. We can now see whether the guard we  proposed is sufficient to
ensure us that no two processes can get the critical section
together. We do that in a simple way, by introducing the (new) formula
$NC \BImast NC$.  If we have $m\satisfies NC\BImast NC$, then we have
$m = m_1 \bullet m_2$, with  $m_1 \satisfies NC$ and $m_2 \satisfies
NC$. This is a fair representation of two processes running the
non-critical section in parallel, each one using a different part of
the memory (cf. the treatment of concurrent composition in
\cite{CMP12,AP16} and  in Concurrent Separation Logic
\cite{OHearn2007}).    

Now consider a process $\alpha_1$ and suppose it has access to the token;
that is, there exists $m_1'$ such that $m_1 \bullet t\sim_{\alpha_1}
m_1'$. If $In$ is valid, then  we have in particular that
$m_1\satisfies In$ and thus we have $m_1' \satisfies C$.  Now, $\alpha_1$
is executing the critical section with $m_1'$. Could another process 
$\alpha_2$ access the critical section with $m_2$? The guard should avoid
it. Indeed, if $Guard$  is valid, then we have $m\satisfies
Guard$. Yet, we have established that $m_1 \bullet t\sim_{\alpha_1}
m_1'$. We also have that $m=m_1 \bullet m_2$ and, by right
composition, we have $m_1\bullet m_2 \bullet t\sim_{\alpha_1} m_1' \bullet
m_2$, thus $m \bullet t\sim_{\alpha_1} m_1' \bullet m_2$. By applying
$m\satisfies Guard$ with $\alpha' = \alpha_1$ and $\alpha'' = \alpha_2$, we have that 
there is no $m'$ such that $m \bullet t\sim_{\alpha_2} m'$. Now, if $\alpha_2$
were to access the critical  section with $m_2$, then we should have
$m_2'$ such that $m_2\bullet t\sim_{\alpha_2} m'_2$.  Then we should have
that $m \bullet t \sim_{\alpha_2} m_2' \bullet m_1$ which would  contradict
what we stated before. Thus $\alpha_2$ cannot enter the critical section.  

However, once in this situation, as we have $m_1' \satisfies C$, we can
use $Out$ to let $\alpha_1$ out of the critical section. As $m_1'\satisfies
Out$, we generate $m_1'\satisfies\BIaneg\BID{\alpha_1}{t}\BIatop$ and
$m_1'\satisfies\BID{\alpha_1}{e}NC$. The first tells us that
there is no $m'$ such that $m'_1\sim_{\alpha_1} m'\bullet t$. But, in our
premiss, we have that $m'_1\sim_{\alpha_1} m_1\bullet t$. Those two facts
are contradictory. Thus, if we want to use this formula, we have to
delete the relation $m'_1\sim_{\alpha_1} m_1\bullet t$. This guarantees  
that $t$ is no longer in $\alpha_1$'s grasp. The second part, 
$m_1'\satisfies\BID{\alpha_1}{e}NC$, gives us a new memory state $m_1''$
such that $m_1'\sim_{\alpha_1}m_1''$ and $m_1''\satisfies NC$. Thus $\alpha_1$
is back in non-critical state. Note that once $m'_1\sim_{\alpha_1}
m_1\bullet t$ is deleted, the guard ceases to be applicable, and
nothing prevents $\alpha_2$ from entering the critical section this time.

\subsection{Evolution in LL, BI, and ERL} \label{subsec:LL-evolution} 

It is perhaps worthwhile pausing at this point to compare the
representation of system evolution that is available here with that
which is available in Linear Logic (LL).  First, we should note that
the nature of the system model employed here is  quite different from
that which would derive from a representation based on LL. Second, in
our setting, as we have explained, we employ a truth-functional
instantiation of the general distributed systems modelling  approach
based on concepts of location,  resource, and process. In the
examples of this paper, the account of process is very limited, being
restricted  to the actions of epistemic agents (with no rich
process-theoretic structure).  Third, as a result of these design
choices, the readily available account of  evolution requires
unpacking the truth-functional semantics, which can be  see in terms
of tableaux proofs (as presented in Section~\ref{sec:tableaux}).
Experience from, for example, Separation Logic \cite{Rey02a} suggests
that  the presence (as in Boolean BI and ERL and ERL$^*$) of a
negation with the  standard classical semantics is a very useful
modelling tool.     

In contrast, representations using LL's  sequent calculus, such as the
logic programming approach described in \cite{Andreoli92,HM94}, employ a less
rich modelling  perspective --- restricted to proofs of sequences of
resource manipulations --- but then give a very direct operational
reading of evolution in this restricted setting. A proof-theoretic
treatment of some underlying ideas in LL may be found in \cite{CY2019}. 
Note, however, that BI includes MILL as a fragment (as we have seen)
and that the basic  propositional systems for BI can be presented as
sequent calculi with well-understood  relationships with LL. Within
the multiplicative fragment of BI, the  same readings of resource
evolution can, of course, be obtained --- we do not consider it
worthwhile to rehearse these readings in the context of our examples,
which are  intended to illustrate resource semantics. We conjecture,
therefore, that it is possible  to give (perhaps labelled) sequent
calculi for ERL and ERL$^*$ that would provide a similar operational
reading of evolution (see the remarks at the 
beginning of Section~\ref{sec:tableaux}) to that which is available in
LL or the  multiplicative fragment of BI.    

To set up a precise correspondence between these evolutions and the
semantic  representation of resource is an interesting issue.  

A brief comparison with `epistemic linear logic' \cite{Garg06} ---
which is about  modelling access control in LL --- is perhaps also
worthwhile. Again,  this work benefits from the syntactic structures
of LL as basis for representing evolution in the setting of the
restricted model of systems that is naturally treated syntactically by
LL. Again, in contrast, we begin from a more comprehensive systems
semantics  --- which accommodates a very general notion of resource,
including ambient system resources and resources  that are local to
agents --- and  treat similar examples in this restricted
instance. Again, we might expect sequent  calculi for ERL and ERL$^*$
to capture a similar treatment of evolution to that provided by LL.

\section{A tableaux calculus for \ERL} \label{sec:tableaux} 

In this section, we provide a labelled calculus for \ERL\ in the spirit
of the calculi previously developed for \BI\ \cite{Gal05a} and \BBI\
\cite{Lar14a} that are based on labels and label constraints allowing
the capture of the semantics of these logics inside the corresponding
calculus. In the case of \BBI, a specific completeness proof, based on
an oracle, has been developed in \cite{Lar14a}. 

Similar labelled calculi have been proposed also for some modal and
epistemic extensions of \BI\ and \BBI\ \cite{Cour15a,CGP15,Cour15b}.
In these cases, the calculus design, used for \BBI, is applied with
specific labels and constraints issued from a semantic analysis of the
considered logic. In the case of the labelled calculus for \ESL\
\cite{Cour15b}, which is an epistemic extension of \BBI, we deal 
with constraints that are parametrized by agents, but do not handle
the presence of resources in the scope of the modal operators (the
local resources). 

While herein provide a tableaux calculus in the continuation of previous works
on modal bunched logics, we note also that we could design a labelled sequent
calculus for ERL and ERL$^*$ that would also be used to provide an
operational reading of evolution through proof construction as in 
some LL fragments. However, our aim in this section is only to
provide, by applying an approach and some proof methods already
developed for other modal bunched logics, a labelled tableaux calculus
for our logic --- both in order to establish its metatheory and as a
general reasoning tool.    

For the present work, we must introduce labels that correspond
to the local resources embedded in operators. As we shall see, we do that 
through a subset $\Lambda_r$ of labels that is in bijection with the set of
local resources $Res$. Similar techniques have been used with the
logic LSM  \cite{CGP15}, which extends \BBI\ with resource-parametrized S4
modalities.  Likewise, the proofs of soundness and completeness of the
calculus with respect to the semantics introduced in Section~\ref{sec:erl} 
are similar to the ones for \ESL, mainly addressing the need to take
the set $\Lambda_r$ into account. Revisiting the remarks in
Section~\ref{sec:erl} about the possibility of working with a hybrid
semantics and then relating \ERL\ to a hybrid version of \ESL, we remark that the 
design of a hybrid tableau calculus would require some specific work
about  using nominals and formulas to replace labels and constraints --- 
and this replacement introduces more complexity and undermines the strong
links with the resource semantics that is central in our approach. 

First, we introduce labels and constraints that correspond,
respectively, to  resources and to the equality and equivalence
relations on resources and  agents. Next, we develop labelled tableaux
for \ERL. Then, we establish  soundness with respect to the resource
semantics, giving the details of the  proof in the appendix. Finally,
we consider countermodel extraction and  completeness, again giving
the details of the proof in the appendix.

\subsection{Labels and constraints} \label{subsec:constraints} 

We consider a finite set of constants  $\Lambda_r$ such that
$|\Lambda_r|=|Res|-1$. On it we build an infinite  countable set of
(resource) constants $\gamma_r$ such that $\Lambda_r\subset\gamma_r$,
and then $\gamma_r = \Lambda_r \cup \{c_1 ,c_2, \ldots \}$. Concatenation 
of lists is denoted  by $\concatList$; $\lgg \ldd$ denotes the empty list. A 
\emph{resource label} is a word built on $\gamma_r$, where the order of letters 
is not taken into account; that is, a finite multiset  $\gamma_r$ and by $\epsilon$ 
the empty word. For example, $xy$ is the composition of  the resource labels $x$
and $y$. We say that $x$ is a \emph{resource sublabel} of $y$ if  and
only if there exists $z$ such that $x z = y$. The set of resource
sublabels of $x$ is denoted $\subLabelR{x}$. 

We define a function $\lambda : Res \to \Lambda_r$ such that: 
\begin{enumerate}
\item $\lambda(e) = \epsilon$; 
\item for all $r \in Res\backslash\{e\}$, $\lambda(r)\in\Lambda_r$; and 
\item $\lambda$ is injective. 
$r=r'$). 
\end{enumerate}
Note that $\lambda$ is trivially a bijection between  $Res$ and
$\Lambda_r\cup\{\epsilon\}$. 

\begin{definition}[Constraints]
A \emph{resource constraint} is an expression of the form $x \relCr
y$, where $x$ and $y$ are resource labels. An \emph{agent constraint}
is an expression of the form $x \relCag{u} y$, where $x$ and $y$ are
resource labels and $u$ belongs to the set of agents $\setAg$.  
\label{def_constraints}
\end{definition}

A \emph{set of constraints} is any set $\mathcal{C}$ that contains
resource constraints and agent constraints. Let $\mathcal{C}$ be a set
of constraints. The (resource) \emph{domain}  of $\mathcal{C}$ is the
set of all resource sublabels that appear in $\mathcal{C}$; that is,  
\[
\domainR{\mathcal{C}} = 
\bigcup_{x \relCr y \in \mathcal{C}} (\subLabelR{x} \cup
\subLabelR{y}) 
\ \cup \ 
\bigcup_{x \relCag{u} y \in \mathcal{C}} (\subLabelR{x} \cup
\subLabelR{y}) .  
\]

Let $\mathcal{C}$ be a set of constraints. The (resource) \emph{alphabet}
 $\alphabetR{\mathcal{C}}$ of $\mathcal{C}$ is the set of resource
 constants that appear in $\mathcal{C}$. In particular,
 $\alphabetR{\mathcal{C}} = \gamma_r \cap  \domainR{\mathcal{C}}$.
 Now we introduce, in Figure \ref{fig_contraints_closure}, the rules
 for constraint closure that allow us to capture the properties of the
 models into the calculus.   

\begin{figure}[t]
\hrule
\vspace{1mm}
Rules for resource constraints:
{\footnotesize
\begin{center}
    \AxiomC{}
    \RLabel{$\langle \epsilon \rangle$}
    \UnaryInfC{$\epsilon \relCr \epsilon$} 
    \DisplayProof 
  \ \ \ \ \ \ \ \ 
    \AxiomC{$x \relCr y$}
    \RLabel{$\langle s_r \rangle$}
    \UnaryInfC{$y \relCr x$} 
    \DisplayProof 
  \ \ \ \ \ \ \ \ 
    \AxiomC{$xy \relCr xy$}
    \RLabel{$\langle d_r \rangle$}
    \UnaryInfC{$x \relCr x$} 
    \DisplayProof 
  \ \ \ \ \ \ \ \ 
    \AxiomC{$x \relCr y$}
    \AxiomC{$y \relCr z$}
    \RLabel{$\langle t_r \rangle$}
    \BinaryInfC{$x \relCr z$} 
    \DisplayProof 
\end{center}
\begin{center}
    \AxiomC{$x \relCr y$}
    \AxiomC{$yk \relCr yk$}
    \RLabel{$\langle c_r \rangle$}
    \BinaryInfC{$xk \relCr yk$} 
    \DisplayProof 
  \ \ \ \ \ \ \ \ 
    \AxiomC{$x \relCag{u} y$}
    \RLabel{$\langle k_r \rangle$}
    \UnaryInfC{$x \relCr x$} 
    \DisplayProof 
\end{center}
Rules for agent constraints:
\begin{center}
    \AxiomC{$x \relCr x$}
    \RLabel{$\langle  r_a \rangle$}
    \UnaryInfC{$x \relCag{v} x$} 
    \DisplayProof 
  \ \ \ \ \ \ \ \ 
    \AxiomC{$x \relCag{u} y$}
    \RLabel{$\langle  s_a \rangle$}
    \UnaryInfC{$y \relCag{u} x$} 
    \DisplayProof 
  \ \ \ \ \ \ \ \
    \AxiomC{$x \relCag{u}y$}
    \AxiomC{$y \relCag{u} z$}
    \RLabel{$\langle  t_a \rangle$}
    \BinaryInfC{$x \relCag{u} z$} 
    \DisplayProof 
  \ \ \ \ \ \ \ \ 
    \AxiomC{$x \relCag{u} y$}
    \AxiomC{$x \relCr k$}
    \RLabel{$\langle  k_a \rangle$}
    \BinaryInfC{$k \relCag{u} y$} 
    \DisplayProof 
\end{center}}
\caption{Rules for constraint closure (for any $u \in \setAg$)}
\label{fig_contraints_closure}
\vspace{1mm}
\hrule
\end{figure}

\begin{definition}[Closure of constraints]
Let $\mathcal{C}$ be a set of constraints. The closure of
$\mathcal{C}$, denoted $\closure{\mathcal{C}}$, is the least relation
closed under the rules of Figure  \ref{fig_contraints_closure} such
that $\mathcal{C} \subseteq \closure{\mathcal{C}}$. 
\label{def_contraints_closure}
\end{definition}

There are six rules ($\langle \epsilon \rangle$, $\langle s_r \rangle$,
$\langle d_r \rangle$, $\langle t_r \rangle$, $\langle c_r \rangle$,
and $\langle k_r \rangle$) that produce resource constraints and four
rules ($\langle r_a \rangle$, $\langle s_a \rangle$, $\langle  t_a
\rangle$, and $\langle  k_a \rangle$) that produce agent constraints.
We note that $v$, introduced in the rule $\langle r_a \rangle$, must
belong to the set of agents $\setAg$. 

\begin{proposition}
The following rules can be derived from the rules of constraint closure:
\begin{center}
    \AxiomC{$xk \relCr y$}
    \RLabel{$\langle p_l \rangle$}
    \UnaryInfC{$x \relCr x$} 
    \DisplayProof 
  \ \ \ \ \ \ \ \ 
    \AxiomC{$x \relCr yk$}
    \RLabel{$\langle p_r \rangle$}
    \UnaryInfC{$y \relCr y$} 
    \DisplayProof 
  \ \ \ \ \ \ \ \ 
    \AxiomC{$xk \relCag{u} y$}
    \RLabel{$\langle q_l \rangle$}
    \UnaryInfC{$x \relCr x$} 
    \DisplayProof 
  \ \ \ \ \ \ \ \ 
    \AxiomC{$x \relCag{u} yk$}
    \RLabel{$\langle  q_r \rangle$}
    \UnaryInfC{$y \relCr y$}    
    \DisplayProof 
    \AxiomC{$x \relCag{u} y$}
    \AxiomC{$x \relCr x'$}
    \AxiomC{$y \relCr y'$}
    \RLabel{$\langle w_a \rangle$}
    \TrinaryInfC{$x' \relCag{u} y'$} 
    \DisplayProof 
\end{center}
\label{prop_new_rules}
\end{proposition}

\begin{corollary}
\label{corollary1}
Let $\mathcal{C}$ be a set of constraints and $u \in \setAg$ be an
agent. 
\begin{enumerate}
\item $x \in \domainR{\closure{\mathcal{C}}}$ iff $x\relCr x \in
  \closure{\mathcal{C}}$ iff $x  \relCag{u}x \in \closure{\mathcal{C}}$. 
\item If $xy \in \domainR{\closure{\mathcal{C}}}$, $x' \relCr x \in
\closure{\mathcal{C}}$, and $y' \relCr y \in \closure{\mathcal{C}}$,
then $xy \relCr x'y' \in \closure{\mathcal{C}}$.  
\label{cor_x_in_domain_EQUIV_x_relCr_x}
\end{enumerate}
\end{corollary}

\begin{proposition}
Let $\mathcal{C}$ be a set of constraints. We have
$\alphabetR{\mathcal{C}} = \alphabetR{\closure{\mathcal{C}}}$. 
\label{prop_alphabet_closure_constraints_equals_alphabet_constraints} 
\end{proposition}

\begin{lemma}[Compactness]
Let $\mathcal{C}$
be a (possibly infinite) set of constraints. 
\begin{enumerate}
\item If $x \relCr y \in \closure{\mathcal{C}}$, 
      then there is a finite set $\mathcal{C}_f$
      such that $\mathcal{C}_f \subseteq \mathcal{C}$ and
      $x \relCr y \in \closure{\mathcal{C}_f}$. 
\item If $x \relCag{u}y \in \closure{\mathcal{C}}$, 
      then there is a finite set $\mathcal{C}_f$
      such that $\mathcal{C}_f \subseteq \mathcal{C}$ and
      $x \relCag{u} y \in \closure{\mathcal{C}_f}$.
      \label{lem_compactness}
\end{enumerate}
\end{lemma}

\subsection{Labelled tableaux for \ERL}

We now define a labelled tableaux calculus for \ERL\ in the
spirit of previous works \cite{Gal05a, Lar14a, Cour15b, DP16} by using
similar definitions and results but based on the specific label and 
contraints definitions. 

\begin{definition} A \emph{labelled formula} is a 3-tuple of the form
$\labelledFormula{S}{\phi}{x}$ such that $S \in \{\mathbb{T},
\mathbb{F}\}$, $\phi \in \lang$ is a formula and $x \in \Lambda_r$ is a 
resource label. A \emph{constrained set of statements} (CSS) is a pair
$\CSS{\mathcal{F}}{\mathcal{C}}$, where $\mathcal{F}$ is a set   of
labelled formulae and $\mathcal{C}$ is a set of constraints,
satisfying the following property, denoted $P_{css}$, 
\begin{center}
if $\labelledFormula{S}{\phi}{x} \in \mathcal{F}$, then
$x \relCr x \in \closure{\mathcal{C}} \ $  ($P_{css}$). 
\end{center}
A CSS $\CSS{\mathcal{F}}{\mathcal{C}}$  is \emph{finite} if $\mathcal{F}$ 
and $\mathcal{C}$ are finite. The relation $\CSSinclusion$ is defined by
$\CSS{\mathcal{F}}{\mathcal{C}} \CSSinclusion
\CSS{\mathcal{F}'}{\mathcal{C}'}$ iff   $\mathcal{F} \subseteq
\mathcal{F}'$ and $\mathcal{C} \subseteq \mathcal{C}'$. We write 
$\CSS{\mathcal{F}_f}{\mathcal{C}_f} \CSSfiniteInclusion
\CSS{\mathcal{F}}{\mathcal{C}}$ when
$\CSS{\mathcal{F}_f}{\mathcal{C}_f} \CSSinclusion 
\CSS{\mathcal{F}}{\mathcal{C}}$ holds and
$\CSS{\mathcal{F}_f}{\mathcal{C}_f}$ is finite,  
meaning that $\mathcal{F}_f$ and $\mathcal{C}_f$ are both finite. 
\label{def_labelled_formula_AND_CSS} 
\end{definition}

\begin{proposition}
For any CSS $\CSS{\mathcal{F}_f}{\mathcal{C}}$, where $\mathcal{F}_f$
is finite, there exists $\mathcal{C}_f \subseteq \mathcal{C}$ such
that $\mathcal{C}_f$ is finite and
$\CSS{\mathcal{F}_f}{\mathcal{C}_f}$ is a CSS. 
\label{prop_css_finite_css}
\end{proposition}

\begin{proof}
By induction on the number of labelled formulae of $\mathcal{F}_f$ and
by Lemma \ref{lem_compactness}. 
\end{proof}

\begin{figure}[htbp]
\hspace{-.2cm}
\center
{\footnotesize 
\begin{tabular}{c}
\hline 
      \\      
      \AxiomC{$\labelledFormula{T}{\BImI} {x} \in \mathcal{F}$}
      \RLabel{$\langle \mathbb{T} \BImI \rangle$} 
      \UnaryInfC{$\CSS{\emptyset}{\{ x \relCr \epsilon \}}$} 
      \DisplayProof 
      
      \\ \\

      \AxiomC{$\labelledFormula{T}{\BIaneg \phi} {x} \in \mathcal{F}$}
      \RLabel{$\langle \mathbb{T} \BIaneg \rangle$} 
      \UnaryInfC{$\CSS{\{ \labelledFormula{F}{\phi} {x} \}}{\emptyset}$}  
      \DisplayProof 
      \ \ \ \ 
      \AxiomC{$\labelledFormula{F}{\BIaneg \phi} {x} \in \mathcal{F}$}
      \RLabel{$\langle \mathbb{F} \BIaneg \rangle$} 
      \UnaryInfC{$\CSS{\{ \labelledFormula{T}{\phi} {x} \}}{\emptyset}$}  
      \DisplayProof 
      
      \\ \\

      \AxiomC{$\labelledFormula{T}{\phi \BIaand \psi} {x} \in \mathcal{F}$} 
      \RLabel{$\langle \mathbb{T} \BIaand \rangle$} 
      \UnaryInfC{$\CSS{\{ \labelledFormula{T}{\phi} {x},
                 \labelledFormula{T}{\psi} {x} \}}{\emptyset}$} 
      \DisplayProof 
      \ \ \ \ 
      \AxiomC{$\labelledFormula{F}{\phi \BIaand \psi} {x} \in \mathcal{F}$} 
      \RLabel{$\langle \mathbb{F} \BIaand \rangle$} 
      \UnaryInfC{$\CSS{\{ \labelledFormula{F}{\phi} {x} \}}{\emptyset}
                 \ \mid \  
                 \CSS{\{ \labelledFormula{F}{\psi} {x} \}}{\emptyset}$}  
      \DisplayProof 
      
      \\ \\

      \AxiomC{$\labelledFormula{T}{\phi \BIaor \psi} {x} \in \mathcal{F}$} 
      \RLabel{$\langle \mathbb{T} \BIaor \rangle$} 
      \UnaryInfC{$\CSS{\{ \labelledFormula{T}{\phi} {x} \}}{\emptyset}
                 \ \mid \  
                 \CSS{\{ \labelledFormula{T}{\psi} {x} \}}{\emptyset}$}  
      \DisplayProof 
      \ \ \ \ 
      \AxiomC{$\labelledFormula{F}{\phi \BIaor \psi} {x} \in \mathcal{F}$} 
      \RLabel{$\langle \mathbb{F} \BIaor \rangle$} 
      \UnaryInfC{$\CSS{\{ \labelledFormula{F}{\phi} {x}, 
                          \labelledFormula{F}{\psi} {x} \}}{\emptyset}$} 
      \DisplayProof 
      
      \\ \\

      \AxiomC{$\labelledFormula{T}{\phi \BIaimp \psi} {x} \in \mathcal{F}$} 
      \RLabel{$\langle \mathbb{T} \BIaimp \rangle$} 
      \UnaryInfC{$\CSS{\{ \labelledFormula{F}{\phi} {x} \}}{\emptyset}
                  \ \mid \  
                  \CSS{\{ \labelledFormula{T}{\psi} {x} \}}{\emptyset}$}  
      \DisplayProof 
      \ \ \ \ 
      \AxiomC{$\labelledFormula{F}{\phi \BIaimp \psi} {x} \in \mathcal{F}$} 
      \RLabel{$\langle \mathbb{F} \BIaimp \rangle$} 
      \UnaryInfC{$\CSS{\{ \labelledFormula{T}{\phi} {x},
                  \labelledFormula{F}{\psi} {x} \}}{\emptyset}$} 
      \DisplayProof 
      
           \\ \\

      \AxiomC{$\labelledFormula{T}{\phi \BImast \psi} {x} \in \mathcal{F}$} 
      \RLabel{$\langle \mathbb{T} \BImast \rangle$} 
      \UnaryInfC{$\CSS{\{ \labelledFormula{T}{\phi}{c_i},
                  \labelledFormula{T}{\psi}{c_j} \}}{\{ x \relCr
                  c_ic_j \}}$}   
      \DisplayProof 
      \ \ \ \ 
      \AxiomC{$\labelledFormula{F}{\phi \BImast \psi} {x} \in \mathcal{F}
               \text{ and } x \relCr yz \in \closure{\mathcal{C}}$}
      \RLabel{$\langle \mathbb{F} \BImast \rangle$} 
      \UnaryInfC{$\CSS{\{ \labelledFormula{F}{\phi}{y} \}}{\emptyset}
                  \ \mid \ \CSS{\{ \labelledFormula{F}{\psi}{z}
                    \}}{\emptyset}$}   
      \DisplayProof 
      
      \\ \\

      \AxiomC{$\labelledFormula{T}{\phi \BImimp \psi} {x} \in \mathcal{F}
               \text{ and } xy \relCr xy \in \closure{\mathcal{C}}$}
      \RLabel{$\langle \mathbb{T} \BImimp \rangle$} 
      \UnaryInfC{$\CSS{\{ \labelledFormula{F}{\phi}{y} \}}{\emptyset}
                  \ \mid \ 
                  \CSS{\{ \labelledFormula{T}{\psi}{xy} \}}{\emptyset}$}  
      \DisplayProof 
      \ \ \ \ 
      \AxiomC{$\labelledFormula{F}{\phi \BImimp \psi} {x} \in \mathcal{F}$} 
      \RLabel{$\langle \mathbb{F} \BImimp \rangle$} 
      \UnaryInfC{$\CSS{\{ \labelledFormula{T}{\phi}{c_i},
                          \labelledFormula{F}{\psi}{xc_i} \}}{\{ xc_i
                          \relCr xc_i \}}$} 

      \DisplayProof 
      
      \\ \\

      \AxiomC{$\labelledFormula{T}{\BIC{u}{r} \phi}{x} \in \mathcal{F}
               \text{ and } x\lambda(r) \relCag{u} y \in
               \closure{\mathcal{C}}$} 
      \RLabel{$\langle \mathbb{T} \BIC{}{} \rangle$} 
      \UnaryInfC{$\CSS{\{ \labelledFormula{T}{\phi}{y} \}}{\emptyset}$}  
      \DisplayProof 
      \ \ \ \ 
      \AxiomC{$\labelledFormula{F}{\BIC{u}{r} \phi}{x} \in \mathcal{F}$}
      \RLabel{$\langle \mathbb{F} \BIC{}{}\rangle$} 
      \UnaryInfC{$\CSS{\{ \labelledFormula{F}{\phi}{c_i} \}}{\{
          x\lambda(r) \relCag{u} c_i \}}$}   
      \DisplayProof 

      \\ \\
 
          \AxiomC{$\labelledFormula{T}{\BID{u}{r} \phi}{x} \in \mathcal{F}$}
      \RLabel{$\langle \mathbb{T} \BID{}{} \rangle$} 
      \UnaryInfC{$\CSS{\{ \labelledFormula{T}{\phi}{c_i \lambda(r)}
          \}}{\{ x \relCag{u} c_i \lambda(r) \}}$}   
      \DisplayProof 
      \ \ \ \ 
      \AxiomC{$\labelledFormula{F}{\BID{u}{r} \phi}{x} \in \mathcal{F}
               \text{ and } x \relCag{u} y\lambda(r) \in
               \closure{\mathcal{C}}$} 
      \RLabel{$\langle \mathbb{F} \BID{}{} \rangle$} 
      \UnaryInfC{$\CSS{\{ \labelledFormula{F}{\phi}{y \lambda(r)}
          \}}{\emptyset}$}   
      \DisplayProof 
        
      \\ \\
      
      \AxiomC{$\labelledFormula{T}{\BIE{u}{r} \phi}{x} \in \mathcal{F}
               \text{ and } x\lambda(r) \relCag{u} y\lambda(r) \in
               \closure{\mathcal{C}}$} 
      \RLabel{$\langle \mathbb{T} \BIE{}{} \rangle$} 
      \UnaryInfC{$\CSS{\{ \labelledFormula{T}{\phi}{y \lambda(r)}
          \}}{\emptyset}$}   
      \DisplayProof 
      \ \ \ \ 
      \AxiomC{$\labelledFormula{F}{\BIE{u}{r} \phi}{x} \in \mathcal{F}$}
      \RLabel{$\langle \mathbb{F} \BIE{}{} \rangle$} 
      \UnaryInfC{$\CSS{\{ \labelledFormula{F}{\phi}{c_i \lambda(r)}
          \}}{\{ x\lambda(r) \relCag{u} c_i \lambda(r) \}}$}   
      \DisplayProof 
         
      \\ \\
      
        \AxiomC{$\labelledFormula{T}{\BICt{u}{r} \phi}{x} \in \mathcal{F}$}
      \RLabel{$\langle \mathbb{T} \BICt{}{} \rangle$} 
      \UnaryInfC{$\CSS{\{ \labelledFormula{T}{\phi}{c_i} \}}{\{
          x\lambda(r) \relCag{u} c_i \}}$}   
      \DisplayProof 
      \ \ \ \      
       \AxiomC{$\labelledFormula{F}{\BICt{u}{r} \phi}{x} \in \mathcal{F}
               \text{ and } x\lambda(r) \relCag{u} y \in
               \closure{\mathcal{C}}$} 
      \RLabel{$\langle \mathbb{F} \BICt{}{} \rangle$} 
      \UnaryInfC{$\CSS{\{ \labelledFormula{F}{\phi}{y} \}}{\emptyset}$}  
      \DisplayProof 
        
      \\ \\

      \AxiomC{$\labelledFormula{T}{\BIDt{u}{r} \phi}{x} \in \mathcal{F}
               \text{ and } x  \relCag{u} y\lambda(r) \in
               \closure{\mathcal{C}}$} 
      \RLabel{$\langle \mathbb{T} \BIDt{}{} \rangle$} 
      \UnaryInfC{$\CSS{\{ \labelledFormula{T}{\phi}{y\lambda(r)}
          \}}{\emptyset}$}   
      \DisplayProof 
      \ \ \ \ 
      \AxiomC{$\labelledFormula{F}{\BIDt{u}{r} \phi}{x} \in \mathcal{F}$}
      \RLabel{$\langle \mathbb{F} \BIDt{}{} \rangle$} 
      \UnaryInfC{$\CSS{\{ \labelledFormula{F}{\phi}{c_i \lambda(r)}
          \}}{\{ x \relCag{u} c_i \lambda(r) \}}$}   
      \DisplayProof 
         
      \\ \\
      
	  \AxiomC{$\labelledFormula{T}{\BIEt{u}{r} \phi}{x} \in \mathcal{F}$}
      \RLabel{$\langle \mathbb{T} \BIEt{}{} \rangle$} 
      \UnaryInfC{$\CSS{\{ \labelledFormula{T}{\phi}{c_i \lambda(r)}
          \}}{\{ x\lambda(r)\relCag{u} c_i \lambda(r) \}}$}   
      \DisplayProof 
       \ \ \ \      
      \AxiomC{$\labelledFormula{F}{\BIEt{u}{r} \phi}{x} \in \mathcal{F}
               \text{ and } x\lambda(r) \relCag{u} y\lambda(r) \in
               \closure{\mathcal{C}}$} 
      \RLabel{$\langle \mathbb{F} \BIEt{}{} \rangle$} 
      \UnaryInfC{$\CSS{\{ \labelledFormula{F}{\phi}{y \lambda(r)}
          \}}{\emptyset}$}   
      \DisplayProof 
         
      \\ \\

      Note: $c_i$ and $c_j$  are new label constants, with $c_i,
      c_j\notin\Lambda_r$.  
\end{tabular}}
\caption{Rules of the tableaux calculus for \ERL
}
\label{fig_tableaux_method_rules}  
\vspace{2mm}
\hrule
\end{figure}

\noindent Figure ~\ref{fig_tableaux_method_rules} presents the rules
of tableaux calculus for \ERL. Note that `$c_i$ and $c_j$ are new
label  constants' means $c_i \not = c_j \in \gamma_r \setminus
(\alphabetR{\mathcal{C}}\cup\Lambda_r)$.   

\begin{definition}[Tableau for \ERL]
Let $\CSS{\mathcal{F}_0}{\mathcal{C}_0}$ be a finite CSS.
A \emph{tableau} for $\CSS{\mathcal{F}_0}{\mathcal{C}_0}$ is a
list of CSSs, called \emph{branches}, inductively built according the
following rules:  
\begin{enumerate}
\item The one branch list $[\CSS{\mathcal{F}_0}{\mathcal{C}_0}]$
      is a tableau for $\CSS{\mathcal{F}_0}{\mathcal{C}_0}$;
\item If the list $\mathcal{T}_m \concatList
[\CSS{\mathcal{F}}{\mathcal{C}}] \concatList \mathcal{T}_n$ 
      is a tableau for $\CSS{\mathcal{F}_0}{\mathcal{C}_0}$ and
\begin{center}
  \AxiomC{cond$\CSS{\mathcal{F}}{\mathcal{C}}$}
  \UnaryInfC{$\CSS{\mathcal{F}_1}{\mathcal{C}_1}$ $\mid$ \ldots $\mid$
$\CSS{\mathcal{F}_k}{\mathcal{C}_k}$}  
  \DisplayProof 
\end{center}
is an instance of a rule of Figure \ref{fig_tableaux_method_rules} for
which cond$\CSS{\mathcal{F}}{\mathcal{C}}$ is fulfilled, then the list
$ \mathcal{T}_m \concatList [\CSS{\mathcal{F} \cup
\mathcal{F}_1}{\mathcal{C} \cup \mathcal{C}_1}; \ldots; 
  \CSS{\mathcal{F} \cup \mathcal{F}_k}{\mathcal{C} \cup
\mathcal{C}_k}] \concatList \mathcal{T}_n$ 
is a tableau for $\CSS{\mathcal{F}_0}{\mathcal{C}_0}$.
\end{enumerate}
A \emph{tableau} for the formula $\phi$ is a \emph{tableau} for
$\CSS{\{ \labelledFormula{F}{\phi}{c_1} \}}{\{ c_1 \relCr c_1 \}}$. 
\label{def_tableau}
\end{definition}

\noindent
We remark that a tableau for a formula $\phi$ verifies the property
($P_{css}$) of Definition \ref{def_labelled_formula_AND_CSS} (by the
rule $\langle r_a \rangle$)  and any application of a rule of Figure
\ref{fig_tableaux_method_rules} provides also  a tableau that verifies
the property ($P_{css}$) (in particular, by Corollary
\ref{corollary1}). 

In this calculus, we have two particular set of rules. The first set
is composed by the rules  
$\langle \mathbb{T} \BImI \rangle$,
$\langle \mathbb{T} \BImast \rangle$,
$\langle \mathbb{F} \BImimp \rangle$,
$\langle \mathbb{F} \BIC{}{} \rangle$,
$\langle \mathbb{F} \BIDt{}{} \rangle$,
$\langle \mathbb{F} \BIE{}{} \rangle$,
$\langle \mathbb{T} \BICt{}{} \rangle$, 
$\langle \mathbb{T} \BID{}{} \rangle$, and
$\langle \mathbb{T} \BIEt{}{} \rangle$,
that introduce new label constants ($c_i$ and $c_j$) and new
constraints, except for $\langle \mathbb{T} \BImI \rangle$ that only
introduces a new constraint. 
The second set is composed of the rules $\langle \mathbb{F} \BImast
\rangle$, $\langle \mathbb{T} \BImimp \rangle$, $\langle \mathbb{T}
\BIC{}{} \rangle$, $\langle \mathbb{T} \BIDt{}{} \rangle$, $\langle
\mathbb{T} \BIE{}{} \rangle$, $\langle \mathbb{F} \BICt{}{} \rangle$,
$\langle \mathbb{F} \BID{}{} \rangle$, and$\langle \mathbb{F}
\BIEt{}{} \rangle$, that have a condition on the closure of
constraints. To apply one of these rules we choose a label which
satisfies the condition and then apply the corresponding
rule. Otherwise, we cannot apply the rule.

\begin{definition}[Closure conditions]
A CSS $\CSS{\mathcal{F}}{\mathcal{C}}$ is \emph{closed} if one of the 
following conditions holds, where $\phi \in \lang$: 
\begin{enumerate}
\item $\labelledFormula{T}{\phi}{x} \in \mathcal{F}$,
$\labelledFormula{F}{\phi}{y} \in \mathcal{F}$ and $x \relCr y \in
\closure{\mathcal{C}}$; 
\item $\labelledFormula{F}{\BImI}{x} \in
\mathcal{F}$ and $x \relCr \epsilon \in \closure{\mathcal{C}}$; 
\item $\labelledFormula{F}{\BIatop}{x} \in \mathcal{F}$; 
\item $\labelledFormula{T}{\BIabot}{x} \in \mathcal{F}$. 
\end{enumerate}
A CSS is \emph{open} if it is not closed. A tableau for $\phi$ is \emph{closed} 
if all its branches (that is, all of its CSSs) are closed and a \emph{tableaux proof} 
for $\phi$ is a closed tableau for $\phi$. 
\label{def_closed_tableau}
\end{definition}

Closed branches are marked with $\times$ and open branches are marked with $\circ$. \\

\noindent{\bf Example}. Let us consider the formula $\BIDt{a}{s}\phi
\BIaimp \BIDt{a}{r}(\BIDt{a}{s} \phi)$. To build the corresponding
tableau, we  start with the CCS
$\CSS{\{ \labelledFormula{F}{\BIDt{a}{s}\phi \BIaimp
    \BIDt{a}{r}(\BIDt{a}{s} \phi)}{c_1} \}}{\{ c_1 \relCr c_1 \}}$ and
with the following representation of the formula set  $\mathcal{F}$
and the constraints set $\mathcal{C}$:  
\[
\begin{array}{c@{\qquad}c} 
[\mathcal{F}] & [\mathcal{C}] \\ 
\surd_1 \labelledFormula{F}{\BIDt{a}{s}\phi 
	\BIaimp \BIDt{a}{r}(\BIDt{a}{s} \phi)}{c_1} & c_1 \relCr c_1 
\end{array}
\]
We then apply the rules of our tableaux method, respecting the priority
order, and we obtain the tableau of Figure \ref{fig:figure_example_tableau}. 
We omit the $\lambda$ and write $r$ for $\lambda(r)$, for any
resource.  

Note that we mark with $\surd$ the steps of the tableau
construction. The main steps are the following: first apply the rule
$\langle \mathbb{F} \BIaimp \rangle$ ($\surd_1$) and then obtain two
formulae both with $\BIDt{}{}$ as operator. According to the 
priority rules, first apply the $\langle \mathbb{F}\BIDt{}{}\rangle$
rule ($\surd_2$), which generates a new formula, a new resource label
$c_2$, and the constraint $c_1 \relCag{a} c_2 r$. Then apply the
$\langle \mathbb{F}\BIDt{}{}\rangle$ rule again ($\surd_3$), which
generates a new formula, a new resource label $c_3$, and the
constraint $c_2 r \relCag{a} c_3 s$. We must now apply the $\langle 
\mathbb{T}\BIDt{}{}\rangle$ rule ($\surd_4$) and then we need a
resource label $z$ such that $c_1 \relCag{a} z s \in
\closure{\mathcal{C}}$. 
 
Now, having closure by rule $\langle t_a\rangle$ with agent $a$, we
generate the constraint $c_1 \relCag{a} c_3 s$, and thus apply the
rule with $z=c_1$ and generate $\labelledFormula{T}{\phi}{c_3 s}$. As
we also have $\labelledFormula{F}{\phi}{c_3 s}$, we have a closed
branch and thus a closed tableau.  

\begin{figure}[th] 
\hrule
 \vspace{1mm}
  \centering
  \begin{tikzpicture}[scale=.75, >=latex]

  \tikzstyle{carre}=[draw,rectangle]
  \tikzstyle{arrete}=[-,thick]
  \tikzstyle{arc}=[->,>=stealth,thick]

  \node at (2,2)(n1_0) {$[\mathcal{F}]$};
 \node at (2,1)(n1_1) {$\surd_1$ $\labelledFormula{F}{\BIDt{a}{s}\phi
     \BIaimp \BIDt{a}{r}(\BIDt{a}{s} \phi)}{c_1}$};  
  \node at (2,0)(n1_2) {$\surd_4$
    $\labelledFormula{T}{\BIDt{a}{s}\phi}{c_1}$};  
  \node at (2,-0.5)(n1_3) {$\surd_2$
    $\labelledFormula{F}{\BIDt{a}{r}(\BIDt{a}{s} \phi)}{c_1}$};  
  \node at (2,-1.5)(n1_4) {$\surd_3$ $\labelledFormula{F}{\BIDt{a}{s}
      \phi}{c_2 r}$};  
  \node at (2,-2.5)(n1_5) {$\labelledFormula{F}{\phi}{c_3 s}$}; 
  \node at (2,-3.5)(n1_6) {$\labelledFormula{T}{\phi}{c_3 s}$}; 
  \node at (2,-4.5)(n1_7) {$\times$}; 

  \draw[arrete] (n1_1) -- (n1_2);
  \draw[arrete] (n1_3) -- (n1_4);
  \draw[arrete] (n1_4) -- (n1_5);
  \draw[arrete] (n1_5) -- (n1_6);
  \draw[arrete] (n1_6) -- (n1_7);

  \node at (8,2)(n2_0) {$[\mathcal{C}]$};
  \node at (8,1)(n2_1) {$c_1 \relCr c_1$};
  \node at (8,-1.5)(n2_2) {$c_1 \relCag{a} c_2 r$};
  \node at (8,-2.5)(n2_3) {$c_2 r \relCag{a} c_3 s$};
  
  \draw[arrete] (n2_1) -- (n2_2);
  \draw[arrete] (n2_2) -- (n2_3);
  \draw[arrete] (n2_3) -- (8,-4.5);
   
  \end{tikzpicture}

  \caption{Tableau for $\BIDt{a}{s}\phi \BIaimp \BIDt{a}{r}(\BIDt{a}{s} \phi)$} 
  \label{fig:figure_example_tableau}
 \vspace{1mm} 
 \hrule
\end{figure}

\subsection{Soundness of the calculus} 

We start by proving the soundness property of the tableaux calculus. The 
proof is similar to the soundness proof developed for \BI\ tableaux and some 
recent extensions \cite{Gal05a,Cour15a,Cour15b,DP16}. We remind here 
the key notions and more detailed proofs are given in Appendix~\ref{proof_soundness}.

The main point is the notion of \emph{realizability} of a CSS
$\CSS{\mathcal{F}}{\mathcal{C}}$, meaning that there exists a model
$\model$ and an embedding ($\realizationR{.}$) from the resource
labels to the resource set of  $\model$ such that if
$\labelledFormula{T}{\phi}{x} \in \mathcal{F}$, then $\realizationR{x}
\forcing{\model} \phi$ ,and if $\labelledFormula{F}{\phi}{x} \in
\mathcal{F}$, then $\realizationR{x} \not \forcing{\model} \phi$.

\begin{definition}[Realization]
Let $\CSS{\mathcal{F}}{\mathcal{C}}$ be a CSS. A \emph{realization} of
it is a pair $(\mathcal{M},\realizationR{.})$ where
$\mathcal{M}=(\mathcal{R},\{\sim_a\}_{a\in A}, V)$ is a model and 
$\realizationR{.}:\domainR{\mathcal{C}}\rightarrow R$ such that 
\begin{itemize}
\item for any $r\in Res$, we have $\realizationR{\lambda(r)}=r$,
\item $\realizationR{\epsilon}=e$,
\item $\realizationR{.}$ is a total function (for all
  $x\in\domainR{\mathcal{C}}$, $\realizationR{x}$ is defined), 
\item if $xy\in\domainR{\mathcal{C}}$, then
  $\realizationR{x}\bullet\realizationR{y}\downarrow$ and
  $\realizationR{x}\bullet\realizationR{y}=\realizationR{xy}$, 
\item if $\labelledFormula{T}{\phi}{x}\in\mathcal{F}$, then
  $\realizationR{x}\satisfies \phi$, 
\item if $\labelledFormula{F}{\phi}{x}\in\mathcal{F}$, then
  $\realizationR{x}\not\satisfies \phi$, 
\item if $x\relCr y\in \mathcal{C}$, then $\realizationR{x}=\realizationR{y}$, and 
\item if $x\relCag{u} y \in \mathcal{C}$, then
  $\realizationR{x}\sim_u\realizationR{y}$. 
\end{itemize}

\label{def_realization}
\end{definition}

We say that a CSS is \emph{realizable} if there exists a realization
of this CSS. We say that a tableau is \emph{realizable} if at least
one of its branches is realizable. 

\begin{proposition}
Let $\CSS{\mathcal{F}}{\mathcal{C}}$ be a CSS and
$\mathcal{R}=(\mathcal{M},\realizationR{.})$ be a realization of
it. $\mathcal{R}$ is also a realization of
$\CSS{\mathcal{F}}{\closure{\mathcal{C}}}$, and then 
\begin{enumerate}
\item for all $x\in\domainR{\closure{\mathcal{C}}}$,
  $\realizationR{x}$ is defined,  
\item if $x\relCr y\in \closure{\mathcal{C}}$, then
  $\realizationR{x}=\realizationR{y}$, and  
\item if $x\relCag{u} y\in\closure{\mathcal{C}}$, then
  $\realizationR{x}\sim_u\realizationR{y}$. 
\end{enumerate}
\label{prop_cloture_realisable}
\end{proposition}

\begin{lemma}
The rules of the tableaux method for \ERL\ preserve realizability 
\label{lem_rules_preserve_realizability}
\end{lemma}

\begin{proof} 
By induction on the structure of realizable tableaux. See \cite{CGP15} 
for a similar argument and Appendix \ref{proof_soundness}
for more details.
\end{proof}

\begin{lemma}
Closed branches are not realizable.
\label{lem_closed_tableau_not_realizable}
\end{lemma}
\begin{proof}
By a case analysis of closed branches that are realizable. 
See \cite{CGP15} for a similar argument and Appendix
\ref{proof_soundness} for more details. 
\end{proof}

\begin{theorem}[Soundness] 
Let $\phi$ be a formula of \ERL. If there exists a tableaux proof for
$\phi$, then $\phi$ is valid.  
\label{th_soundness} 
\end{theorem}

\begin{proof}
We suppose that there exists a proof for $\phi$. Then there is a
closed tableau $\mathcal{T}_\phi$ for the CSS $\mathfrak{C} = \CSS{\{
  \labelledFormula{F}{\phi}{c_1}\}}{\{ c_1   \relCr c_1 \}}$. 
Now suppose that $\phi$ is not valid. Then there is a
countermodel $\model = (\PRM, \{ \relR{a} \}_{a \in \setAg}, \val)$
and a resource $r \in \setR$ such that $r \not \satisfies \phi$. 
Let $\mathfrak{R} = (\model, \realizationR{.})$ such that
$\realizationR{c_1} = r$. As $\mathfrak{R}$ is a realization of
$\mathfrak{C}$, by Lemma \ref{lem_rules_preserve_realizability},
$\mathcal{T}_\phi$ is realizable. Moreover by Lemma
\ref{lem_closed_tableau_not_realizable}, $\mathcal{T}_\phi$ cannot be
closed, which is absurd because $\mathcal{T}_\phi$ is a proof and then
is closed by definition. Therefore $\phi$ is valid.
\end{proof}

\subsection{Countermodel generation and Completeness of the calculus} 

Before proceeding to establish completeness, we consider 
a countermodel extraction method for our calculus that is adapted from
a method proposed in \cite{Lar14a}. \\

\noindent
{\bf Countermodel generation}. The method transforms the sets of
resource and agent constraints of a branch
$\CSS{\mathcal{F}}{\mathcal{C}}$ into a  model $\model$ such that, if 
$\labelledFormula{T}{\phi}{x} \in \mathcal{F}$, then $\rho_x
\forcing{\model} \phi$ and, if  $\labelledFormula{F}{\phi}{x} \in
\mathcal{F}$, then $\rho_x \not 
\forcing{\model} \phi$, where $\rho_x$ is the representative of the
equivalence class of $x$. 

The method is based mainly on the definition on a particular CSS
$\CSS{\mathcal{F}}{\mathcal{C}}$, called a \emph{Hintikka CSS}. For
more details, see Appendix \ref{sec:extraction}. This approach for
countermodel extraction is proposed and illustrated for other bunched
logics in \cite{Gal05a,Cour15a,Cour15b,CGP15,DP16} and adapted to our
\ERL\ logic. 

\medskip

\noindent {\bf Example}. 
We give an example of countermodel extraction by considering
$A=\{a\}$ and $Res=\{e,r\}$ and the formula
$\BIC{a}{s}\phi\BIaimp\BIC{a}{r}\BIC{a}{s} \phi$, which is not valid. 
By applications of the tableaux rules, we obtain the tableau of 
Fig  \ref{figure_example_tableau}. 

\begin{figure}[th]
\hrule 
\vspace{1mm}
  \centering
  \begin{tikzpicture}[scale=1, >=latex], show background rectangle] 
                                <=== defini l'echelle 

  \tikzstyle{carre}=[draw,rectangle]
  \tikzstyle{arrete}=[-,thick]
  \tikzstyle{arc}=[->,>=stealth,thick]

  \node at (2,2)(n1_0) {$[\mathcal{F}]$};
 \node at (2,1)(n1_1) {$\surd_1$ $\labelledFormula{F}{\BIC{a}{s}\phi
     \BIaimp \BIC{a}{s}(\BIC{a}{r} \phi)}{c_1}$};  
  \node at (2,0)(n1_2) {$\surd_4$
    $\labelledFormula{T}{\BIC{a}{s}\phi}{c_1}$};  
  \node at (2,-0.5)(n1_3) {$\surd_2$
    $\labelledFormula{F}{\BIC{a}{s}(\BIC{a}{r} \phi)}{c_1}$};  
  \node at (2,-1.5)(n1_4) {$\surd_3$ $\labelledFormula{F}{\BIC{a}{r}
      \phi}{c_2}$};  
  \node at (2,-2.5)(n1_5) {$\labelledFormula{F}{\phi}{c_3}$}; 
  \node at (2,-3.5)(n1_6) {$\labelledFormula{T}{\phi}{c_2}$}; 
  \node at (2,-4.5)(n1_7) {$\Circle$}; 

  \draw[arrete] (n1_1) -- (n1_2);
  \draw[arrete] (n1_3) -- (n1_4);
  \draw[arrete] (n1_4) -- (n1_5);
  \draw[arrete] (n1_5) -- (n1_6);
  \draw[arrete] (n1_6) -- (n1_7);

  \node at (8,2)(n2_0) {$[\mathcal{C}]$};
  \node at (8,1)(n2_1) {$c_1 \relCr c_1$};
  \node at (8,-1.5)(n2_2) {$c_1 s \relCag{a} c_2$};
  \node at (8,-2.5)(n2_3) {$c_2 r \relCag{a} c_3$};
  
  \draw[arrete] (n2_1) -- (n2_2);
  \draw[arrete] (n2_2) -- (n2_3);
  \draw[arrete] (n2_3) -- (8,-4.5);
 
 \end{tikzpicture}

\caption{Tableau for $\BIC{a}{s}\phi \BIaimp \BIC{a}{s}(\BIC{a}{r} \phi)$} 
\label{figure_example_tableau}
\vspace{1mm}
\hrule
\end{figure}

We see that, in step 4, we can only find $c_2$ as suitable label for
$c_1 s\relCag{a} x$ and thus the tableau is not closed. The only
branch of this tableau is a Hintikka CSS and we extract this
countermodel using Definition \ref{def_function_omega}. 

We have $\model=(\PRM, \{ \relR{a} \}_{a \in \setAg}, \val)$, where  
\begin{itemize}
\item $R = Rep(\domainR{\closure{\mathcal{C}}})\cup Res =
  \{e,r,s,\rho_{c_1},\rho_{c_2},\rho_{c_3},\rho_{c_1\lambda(s)},\rho_{c_2\lambda(r)}\}$     
\item The resource composition:
 \begin{center}
      \begin{tabular}{|c||c|c|c|c|c|c|c|c|}
      \hline
      $\compR$ & $\unitR$ & $r$ & $s$ & $\rho_{c_1}$ & $\rho_{c_2}$ &
      $\rho_{c_3}$ & $\rho_{c_1\lambda(s)}$ & $\rho_{c_2\lambda(r)}$.\\
      
      \hline \hline
      $\unitR$ & $\unitR$ & $r$ & $s$ & $\rho_{c_1}$ & $\rho_{c_2}$ &
      $\rho_{c_3}$ & $\rho_{c_1\lambda(s)}$ & $\rho_{c_2\lambda(r)}$\\
      
      \hline
      $r$  & $r$ & $\uparrow$ & $\uparrow$ & $\uparrow$ &
      $\rho_{c_2\lambda(r)}$& $\uparrow$ & $\uparrow$ & $\uparrow$ \\ 
      \hline
      $s$ & $s$ & $\uparrow$ & $\uparrow$ & $\rho_{c_1\lambda(s)}$ &
      $\uparrow$ & $\uparrow$ & $\uparrow$ & $\uparrow$  \\ 
      \hline
       $\rho_{c_1}$  & $\rho_{c_1}$  &  $\uparrow$ &
       $\rho_{c_1\lambda(s)}$ & $\uparrow$ & $\uparrow$ & $\uparrow$ &
       $\uparrow$ & $\uparrow$  \\ 
      \hline
      $\rho_{c_2}$  & $\rho_{c_2}$  &  $\rho_{c_2\lambda(r)}$ &
      $\uparrow$ & $\uparrow$ & $\uparrow$ & $\uparrow$ & $\uparrow$ &
      $\uparrow$  \\ 
      \hline
       $\rho_{c_3}$  & $\rho_{c_3}$  &  $\uparrow$ & $\uparrow$ &
       $\uparrow$ & $\uparrow$ & $\uparrow$ & $\uparrow$ & $\uparrow$
       \\ 
      \hline
      $\rho_{c_1\lambda(s)}$& $\rho_{c_1\lambda(s)}$ &  $\uparrow$ &
      $\uparrow$ & $\uparrow$ & $\uparrow$ & $\uparrow$ & $\uparrow$ &
      $\uparrow$  \\ 
      \hline
       $\rho_{c_2\lambda(r)}$ & $\rho_{c_2\lambda(r)}$ & $\uparrow$ &
       $\uparrow$ & $\uparrow$ & $\uparrow$ & $\uparrow$ & $\uparrow$
       & $\uparrow$  \\ 
      \hline
      \end{tabular}
      \end{center}   
\item The equivalence relation, reflexivity is not represented:
         \begin{center}
            \begin{tikzpicture}[scale=1, >=latex], show background
                                rectangle]  <=== defini l'echelle 
              \tikzstyle{carre}=[draw,rectangle]
              \tikzstyle{arrete}=[-,thick]
              \tikzstyle{arc}=[<->,>=stealth,thick]

              \node[carre] at (1.5,1)(n1) {$e$};
              \node[carre] at (1.5,0)(n1) {$r$};
              \node[carre] at (1.5,-1)(n1) {$s$};
              \node[carre] at (3,1)(n4) {$\rho_{c_1\lambda(s)}$};
              \node[carre] at (5,1)(n5) {$\rho_{c_2}$};
              \draw[arc] (n4) -- (n5) node[midway,above]{$a$};
              \node[carre] at (3,0)(n4b) {$\rho_{c_2\lambda(r)}$};
              \node[carre] at (5,0)(n5b) {$\rho_{c_3}$};
              \draw[arc] (n4b) -- (n5b) node[midway,above]{$a$};
              \node[carre] at (3,-1)(n1) {$\rho_{c_1}$};
            \end{tikzpicture}
          \end{center}
 \item $V(\phi) = \{ \rho_{c_2} \}$. 
\end{itemize}

We can easily verify that we have a countermodel of $\BIC{a}{s}\phi
\BIaimp \BIC{a}{s}(\BIC{a}{r} \phi)$. 
\begin{enumerate}
\item As $\rho_{c_2}\in V(\phi)$, we have $\rho_{c_2}\models\phi$.
\item As $\{x\in R|\rho_{c_1}\bullet s\sim_a x\}=\{\rho_{c_2}\}$, we
  have by (1), $\rho_{c_1}\satisfies\BIC{a}{s}\phi$. 
\item As $\rho_{c_3} \notin V(\phi)$, we have $\rho_{c_3}\not\models\phi$.
\item As $\rho_{c_2} \bullet r=\rho_{c_2\lambda(r)}\sim_a \rho_{c_3}$,
  by (3), we have $\rho_{c_2}\not\satisfies\BIC{a}{r}\phi$. 
\item As $\rho_{c_1}\bullet s=\rho_{c_1\lambda(s)}\sim_a \rho_{c_2}$,
  by (4), we have $\rho_{c_1}\not\satisfies\BIC{a}{s}(\BIC{a}{r}\phi)$. 
\item By (2) and (5), we conclude that
  $\rho_{c_1}\not\satisfies\BIC{a}{s}\phi\BIaimp\BIC{a}{s}(\BIC{a}{r}\phi)$. 
\end{enumerate} 

\noindent 
{\bf Completeness.} 
The proof of completeness  is an extension of the corresponding proof
proposed for \BBI\ \cite{Lar14a} to the epistemic connectives of our
logic. It consists in building, using a fair strategy, a Hintikka CSS from 
a formula for which there is no tableaux proof that is a
sequence of labelled formulae in which all labelled formulae occur
infinitely many times, and also an oracle that is a set of non-closed
CSS with some specific properties.  Then, assuming there is no
tableaux proof for $\phi$, we build a Hintikka CSS, and deduce from it
that $\phi$ is not valid. 

\begin{theorem}[Completeness]
Let $\phi$ be an \ERL\ formula. If $\phi$ is valid, then there exists a
tableaux proof for $\phi$.  
\label{th_completeness}
\end{theorem}
\begin{proof}
The proof is an extension of the corresponding proof proposed for \BBI\
\cite{Lar14a} to the epistemic connectives of our logic. More details
are given in Appendix \ref{sec:proof-of-completeness}. 
\end{proof}


To complete this section, we show how we can define a tableaux calculus
for the sublogic \ERLast\ .

\begin{definition}[Tableaux for \ERLast] \label{def_erlast_tableaux}
The tableaux calculus for \ERLast\ is defined exactly as the tableaux
calculus for \ERL, with the addition of the following rule to
Definition \ref{def_contraints_closure}: 

\begin{center}
\AxiomC{$x \relCag{u} y$}
    \AxiomC{$yk \relCr yk$}
    \RLabel{$\langle c_a \rangle$}
    \BinaryInfC{$xk \relCag{u} yk$} 
    \DisplayProof
\end{center}

\end{definition}

\begin{proposition} \label{prop_soundness_completeness_erlast}
The tableaux calculus for \ERLast\ is sound and complete with respect to the semantics 
given in Sections \ref{sec:erl} and \ref{sec:properties}.
\end{proposition}
\begin{proof}
The proof is the same as the one for \ERL \ except that the new rule $\langle c_a \rangle$ 
must be considered each time the closure of constraints is concerned. This addition does 
not cause any difficulties with proofs since this rule is a direct translation of the specific 
property of \ERLast\ as described in Definition \ref{def_erl*}. 
\end{proof}

\section{Conclusions} \label{sec:conclusions}

We have presented a substructural epistemic logic, based on Boolean
BI, in which the epistemic modalities, which extend the usual
epistemic modalities, are parametrized on the agent's local
resource. The logic represents a first step in developing an epistemic
resource semantics. This  step is illustrated through  examples that
explore the gap between policy and implementation in access
control. We have also provided a system of labelled tableaux for the
logic, and established soundness and completeness. 

Much further work is suggested. First, the theory, pragmatics, and
interpretation of the epistemic modalities with resource semantics,
including  aspects of local reasoning  for resource-carrying agents
\cite{OHea01a,Rey02a}, concurrency  \cite{OHearn2007}. Second, logical
theory, including proof systems, model-theoretic  properties, and
complexity. Connections with other approaches to modelling the
relationship  between policy and implementation in system management,
such as those discussed in  \cite{Caires10} and approaches involving
logics for layered graphs \cite{AP16,CMP15}  should be explored.

\subsection*{Acknowledgements} We are grateful to Simon Docherty 
and to the anonymous referees for their comments on earlier drafts 
of this paper. We also thank many colleagues, including particularly 
James Brotherston, Johan van Benthem, and Peter O'Hearn, among 
many, who have commented on documents related to this document.

\appendix

\section{Soundness: proofs of lemmas}
\label{proof_soundness}

{\bf Lemma 3.} 
The rules of the tableaux method for \ERL\ preserve realizability. 

\begin{proof} 
By induction on the structure of realizable tableaux. See \cite{CGP15} 
for a similar argument. 
Let $\mathcal{T}$ be a realizable tableau.
By definition, $\mathcal{T}$ has a realizable branch
$\mathcal{B} = \CSS{\mathcal{F}}{\mathcal{C}}$.
Let $\mathfrak{R} = (\model, \realizationR{.})$ be a realization of
the branch $\mathcal{B}$, where 
$\model = (\PRM, \{ \relR{a} \}_{a \in \setAg}, \val)$
and
$\realizationR{.} : \domainR{\mathcal{C}} \rightarrow \setR$.
If we apply a rule on a labelled formula of a branch that is not
$\mathcal{B}$ then $\mathcal{B}$ is not modified, and then
$\mathcal{T}$ is realizable. 
Else, we consider each kind of formula on which the rule is applied. 
\begin{itemize}
  \item $\labelledFormula{T}{\BImI}{x} \in \mathcal{F}$. \\
        We have, by definition of realization,
        $\realizationR{x} \satisfies\BImI$.
        Then $\realizationR{x} = \unitR$.
        As $\realizationR{\epsilon} = \unitR$
        then $\realizationR{x} = \realizationR{\epsilon}$
        and we remark that 
        $\mathfrak{R}$ is a realization of
        the new branch
        $\CSS{\mathcal{F}}{\mathcal{C} \cup \{ x \relCr \epsilon \}}$.

  \item $\labelledFormula{T}{\phi_1 \BImast \phi_2}{x} \in \mathcal{F}$. \\
        By realization, we have
        $\realizationR{x} \satisfies \phi_1 \BImast \phi_2$.
        Then, by definition, there exist $r_1, r_2 \in \setR$ such that
        $r_1 \compR r_2 \downarrow$,
        $\realizationR{x} = r_1 \compR r_2$,
        $r_1 \satisfies \phi_1$ and
        $r_2 \satisfies \phi_2$.
        As $c_i$ and $c_j$ are new resource label constants,
        $\realizationR{c_i}$ and $\realizationR{c_j}$ are not defined.
        Moreover as $c_i \not = c_j$,
        we can extend $\mathfrak{R}$ by setting
        $\realizationR{c_i} = r_1$ and $\realizationR{c_j} = r_2$.
        As we have $\realizationR{c_i} \compR \realizationR{c_j}
        \downarrow$ and, by implicit extension,
        $\realizationR{x} = \realizationR{c_i} \compR
        \realizationR{c_j} = \realizationR{c_ic_j}$, 
        we obtain a realization of 
        $\CSS{\mathcal{F}}{\mathcal{C} \cup \{ x \relCr c_ic_j \}}$, 
        that is  a realization of the branch
        $\CSS{\mathcal{F} \cup
        \{ \labelledFormula{T}{c_i},
        \labelledFormula{T}{\phi_2}{c_j} \}}
        {\mathcal{C} \cup \{ x \relCr c_ic_j \}}$.
  
  \item $\labelledFormula{F}{\phi_1 \BImast \phi_2}{x} \in \mathcal{F}$. \\
        We have $\realizationR{x} \not \satisfies \phi_1 \BImast \phi_2$.
        By definition, for all $r_1, r_2 \in \setR$ such that
        $r_1 \compR r_2 \downarrow$ and
        $\realizationR{x} = r_1 \compR r_2$,
        we have $r_1 \not \satisfies \phi$ or
        $r_2 \not \satisfies \psi$.
        The branch is expanded into two branches that are
        $\CSS{\mathcal{F} \cup \{ \labelledFormula{F}{\phi}{y}
          \}}{\mathcal{C}}$ and 
        $\CSS{\mathcal{F} \cup \{ \labelledFormula{F}{\psi}{z}
          \}}{\mathcal{C}}$,  
        where $x \relCr yz \in \closure{\mathcal{C}}$.
        By Proposition \ref{prop_cloture_realisable},
        $\realizationR{x} = \realizationR{yz}$.
        By definition of realization,
        $\realizationR{.}$ is total, then 
        $\realizationR{y} \compR \realizationR{z} \downarrow$ and
        $\realizationR{yz} = \realizationR{y} \compR \realizationR{z}$.
        Thus $\realizationR{y} \not \satisfies \phi$
        or $\realizationR{z} \not \satisfies \psi$.
        Therefore $\mathfrak{R}$ is a realization of at
        least one of the two new branches
        $\CSS{\mathcal{F} \cup \{ \labelledFormula{F}{\phi}{y}
          \}}{\mathcal{C}}$ 
        or $\CSS{\mathcal{F} \cup \{ \labelledFormula{F}{\psi}{z}
          \}}{\mathcal{C}}$. 

  \item $\labelledFormula{T}{\BIC{u}{r}\phi}{x}\in\mathcal{F}$ and
    $x \lambda(r)\relCag{u}y\in\mathcal{\closure{C}}$. \\ 
  We have $\realizationR{x}\satisfies\BIC{u}{r}\phi$. By definition, for
  all $r' \in R$ such that $\realizationR{x}\bullet r\sim_u r'$, we
  have $r'\satisfies \phi$. Moreover, as
  $x\lambda(r)\relCag{u}y\in\mathcal{\closure{C}}$, by Proposition
  \ref{prop_cloture_realisable}, we have
  $\realizationR{x\lambda(r)}\sim_u \realizationR{y}$. By definition,
  $\realizationR{x\lambda(r)}=\realizationR{x}\bullet\realizationR{\lambda(r)}=\realizationR{x}\bullet r$. 
  Thus, $\realizationR{x}\bullet r\sim_u \realizationR{y}$ and
  finally, we have $\realizationR{y}\satisfies\phi$, thus
  $\mathcal{R}$ is a realization of the branch    $\CSS{
    \mathcal{F}\cup\{\labelledFormula{T}{\phi}{y}\}}{\mathcal{C}}$.          
  
  \item $\labelledFormula{F}{\BIC{u}{r}\phi}{x}\in\mathcal{F}$. \\ 
  We have $\realizationR{x}\not\satisfies\BIC{u}{r}\phi$. By definition,
  there exists $r'\in R$ such that $\realizationR{x}\bullet r\sim_u r'$
  and $r'\not\satisfies \phi$. As $c_i$ is a new constraint,
  $\realizationR{c_i}$ is not defined and we can choose
  $\realizationR{c_i}=r'$ and we have $\realizationR{c_i}\not\satisfies
  \phi$ and $\realizationR{x}\bullet r\sim_u \realizationR{c_i}$. By
  definition,
  $\realizationR{x\lambda(r)}=\realizationR{x}\bullet\realizationR{\lambda(r)}=\realizationR{x}\bullet r$. Thus $\realizationR{x\lambda(r)}\sim_u \realizationR{c_i}$ and we have a realization of the branch $\CSS{\mathcal{F}\cup\{\labelledFormula{F}{\phi}{c_i}\}}{\mathcal{C}\cup\{x\lambda(r)\relCag{u}c_i\}}$.   
\end{itemize}
Other cases are proved similarly.
\end{proof}

\noindent
{\bf Lemma 4.} 
Closed branches are not realizable.

\begin{proof}
By a case analysis of closed branches that are realizable. 
See \cite{CGP15} for more details. 

Let $\CSS{\mathcal{F}}{\mathcal{C}}$ a closed branch.
We suppose that this branch is realizable.
Let $\mathfrak{R} = (\model, \realizationR{.})$
a realization of it. There are four cases:
\begin{itemize}
\item $\labelledFormula{T}{\phi}{x}\in\mathcal{F}$,
  $\labelledFormula{F}{\phi}{y}\in\mathcal{F}$ and $x\relCr
  y\in\closure{\mathcal{C}}$. By Proposition
  \ref{prop_cloture_realisable}, as the branch is realizable, we must
  have $\realizationR{x}\satisfies\phi$,
  $\realizationR{y}\not\satisfies\phi$ and
  $\realizationR{x}=\realizationR{y}$, which is absurd. 
\item $\labelledFormula{F}{\BImI}{x}\in\mathcal{F}$ and $x\relCr
  \epsilon\in\closure{\mathcal{C}}$. By Proposition
  \ref{prop_cloture_realisable}, as the branch is realizable, we must
  have $\realizationR{x}\not\satisfies\BImI$ and
  $\realizationR{x}=\realizationR{\epsilon}$. By Definition
  \ref{def:sat-valid}, we have $e\neq\realizationR{x}$ and by
  Definition \ref{def_realization} we have $\realizationR{x}= e$,
  which is absurd. 
\item $\labelledFormula{F}{\BIatop}{x}\in\mathcal{F}$. By Proposition
  \ref{prop_cloture_realisable}, as the branch is realizable, we must
  have $\realizationR{x}\not\satisfies\BIatop$, which is absurd by
  Definition \ref{def:sat-valid}. 
\item $\labelledFormula{T}{\BIabot}{x}\in\mathcal{F}$. By Proposition
  \ref{prop_cloture_realisable}, as the branch is realizable, we must
  have $\realizationR{x}\satisfies\BIabot$, which is absurd by Definition
  \ref{def:sat-valid}. 

\end{itemize}
As all cases are absurd, we conclude that
$\CSS{\mathcal{F}}{\mathcal{C}}$ is not realizable.  
\end{proof}

\section{Countermodel extraction method}
\label{sec:extraction}

We  propose a countermodel extraction method, first designed in
\cite{Lar14a} for \BBI, that consists in transforming the sets of
resource and agent constraints of a branch
$\CSS{\mathcal{F}}{\mathcal{C}}$ into a model $\model$ such that if
$\labelledFormula{T}{\phi}{x} \in \mathcal{F}$ then $\rho_x
\forcing{\model} \phi$ and if  $\labelledFormula{F}{\phi}{x} \in
\mathcal{F}$ then $\rho_x \not \forcing{\model} \phi$, where $\rho_x$
is the representative of the equivalence class of $x$. First, we
define when a CSS $\CSS{\mathcal{F}}{\mathcal{C}}$ is a \emph{Hintikka
  CSS}.   

\begin{definition}[Hintikka CSS]
A CSS $\CSS{\mathcal{F}}{\mathcal{C}}$ is a \emph{Hintikka CSS} iff, 
for any formula $\phi, \psi\in\lang$, any resource $r\in Res$, any
resource label $x,y,z\in \Lambda_r$, and any agent $u\in A$: 
{\footnotesize\small
\begin{enumerate}\setlength{\itemsep}{0mm}
\item $\labelledFormula{T}{\phi}{x}\notin\mathcal{F}$ or
  $\labelledFormula{F}{\phi}{y}\notin\mathcal{F}$ or $x \relCr y
  \notin \closure{\mathcal{C}}$
\item $\labelledFormula{F}{\BImI}{x} \notin \mathcal{F}$ or $x
\relCr \epsilon \notin \closure{\mathcal{C}}$ 
\item $\labelledFormula{F}{\BIatop}{x} \notin \mathcal{F}$
\item $\labelledFormula{T}{\BIabot}{x} \notin \mathcal{F}$
\item If $\labelledFormula{T}{\BImI}{x}\in\mathcal{F}$, then
  $x\relCr\epsilon\in\closure{\mathcal{C}}$ 
\item If $\labelledFormula{T}{\BIaneg\phi}{x}\in\mathcal{F}$, then
  $\labelledFormula{F}{\phi}{x}\in\mathcal{F}$ 
\item If $\labelledFormula{F}{\BIaneg\phi}{x}\in\mathcal{F}$, then
  $\labelledFormula{T}{\phi}{x}\in\mathcal{F}$ 
\item If $\labelledFormula{T}{\phi\BIaand\psi}{x}\in\mathcal{F}$, then
  $\labelledFormula{T}{\phi}{x}\in\mathcal{F}$ and
  $\labelledFormula{T}{\psi}{x}\in\mathcal{F}$ 
\item If $\labelledFormula{F}{\phi\BIaand\psi}{x}\in\mathcal{F}$, then
  $\labelledFormula{F}{\phi}{x}\in\mathcal{F}$ or
  $\labelledFormula{F}{\psi}{x}\in\mathcal{F}$ 
\item If $\labelledFormula{T}{\phi\BIaor\psi}{x}\in\mathcal{F}$, then
  $\labelledFormula{T}{\phi}{x}\in\mathcal{F}$ or
  $\labelledFormula{T}{\psi}{x}\in\mathcal{F}$ 
\item If $\labelledFormula{F}{\phi\BIaor\psi}{x}\in\mathcal{F}$, then
  $\labelledFormula{F}{\phi}{x}\in\mathcal{F}$ and
  $\labelledFormula{F}{\psi}{x}\in\mathcal{F}$ 
\item If $\labelledFormula{T}{\phi\BIaimp\psi}{x}\in\mathcal{F}$, then
  $\labelledFormula{F}{\phi}{x}\in\mathcal{F}$ or
  $\labelledFormula{T}{\psi}{x}\in\mathcal{F}$ 
\item If $\labelledFormula{F}{\phi\BIaimp\psi}{x}\in\mathcal{F}$, then
  $\labelledFormula{T}{\phi}{x}\in\mathcal{F}$ and
  $\labelledFormula{F}{\psi}{x}\in\mathcal{F}$ 
\item If $\labelledFormula{T}{\phi\BImast\psi}{x}\in\mathcal{F}$, then
  $\exists y,z\in \Lambda_r$, $x\relCr yz\in\closure{\mathcal{C}}$ and
  $\labelledFormula{T}{\phi}{y}\in\mathcal{F}$ and
  $\labelledFormula{T}{\psi}{z}\in\mathcal{F}$
\item If $\labelledFormula{F}{\phi\BImast\psi}{x}\in\mathcal{F}$, then
  $\forall y,z\in \Lambda_r$, $x\relCr yz\in\closure{\mathcal{C}}$ implies
  $\labelledFormula{F}{\phi}{y}\in\mathcal{F}$ or
  $\labelledFormula{F}{\psi}{z}\in\mathcal{F}$ 
\item If $\labelledFormula{T}{\phi\BImimp\psi}{x}\in\mathcal{F}$, then
  $\forall y\in \Lambda_r$, $xy \in \mathcal{D}_r$ implies
  $\labelledFormula{F}{\phi}{y}\in\mathcal{F}$ or
  $\labelledFormula{T}{\psi}{xy}\in\mathcal{F}$ 
\item If $\labelledFormula{F}{\phi\BImimp\psi}{x}\in\mathcal{F}$, then
  $\exists y\in \Lambda_r$, $xy \in \mathcal{D}_r$ and
  $\labelledFormula{T}{\phi}{y}\in\mathcal{F}$ and
  $\labelledFormula{F}{\psi}{xy}\in\mathcal{F}$ 
\item If $\labelledFormula{T}{\BIC{u}{r}\phi}{x}\in\mathcal{F}$, then
  $\forall y\in \Lambda_r$, $x\lambda(r)\relCag{u}
  y\in\closure{\mathcal{C}}$ implies
  $\labelledFormula{T}{\phi}{y}\in\mathcal{F}$ 
\item If $\labelledFormula{F}{\BIC{u}{r}\phi}{x}\in\mathcal{F}$, then
  $\exists y\in \Lambda_r$, $x\lambda(r)\relCag{u}
  y\in\closure{\mathcal{C}}$ and
  $\labelledFormula{F}{\phi}{y}\in\mathcal{F}$ 
\item If $\labelledFormula{T}{\BID{u}{r}\phi}{x}\in\mathcal{F}$, then
  there exists $y\in \Lambda_r$, $x\relCag{u}
  y\lambda(r)\in\closure{\mathcal{C}}$ and
  $\labelledFormula{T}{\phi}{y\lambda(r)}\in\mathcal{F}$
\item If $\labelledFormula{F}{\BID{u}{r}\phi}{x}\in\mathcal{F}$, then
  for all $y \in \Lambda_r$, $x\relCag{u}
  y\lambda(r)\in\closure{\mathcal{C}}$ implies
  $\labelledFormula{F}{\phi}{y\lambda(r)}\in\mathcal{F}$ 
\item If $\labelledFormula{T}{\BIE{u}{r}\phi}{x}\in\mathcal{F}$, then
  for all $y\in \Lambda_r$, $x\lambda(r)\relCag{u}
  y\lambda(r)\in\closure{\mathcal{C}}$ implies
  $\labelledFormula{T}{\phi}{y\lambda(r)}\in\mathcal{F}$ 
\item If $\labelledFormula{F}{\BIE{u}{r}\phi}{x}\in\mathcal{F}$, then
  there exists $y\in \Lambda_r$, $x\lambda(r)\relCag{u}
  y\lambda(r)\in\closure{\mathcal{C}}$ and
  $\labelledFormula{F}{\phi}{y\lambda(r)}\in\mathcal{F}$ 
\item If $\labelledFormula{T}{\BICt{u}{r}\phi}{x}\in\mathcal{F}$, then
  there exists $y \in \Lambda_r$, $x\lambda(r)\relCag{u}
  y\in\closure{\mathcal{C}}$ and
  $\labelledFormula{T}{\phi}{y}\in\mathcal{F}$ 
\item If $\labelledFormula{F}{\BICt{u}{r}\phi}{x}\in\mathcal{F}$, then
  for all $y \in \Lambda_r$, $x\lambda(r)\relCag{u}
  y\in\closure{\mathcal{C}}$ implies
  $\labelledFormula{F}{\phi}{y}\in\mathcal{F}$ 
\item If $\labelledFormula{T}{\BIDt{u}{r}\phi}{x}\in\mathcal{F}$, then
  $\forall y\in \Lambda_r$, $x\relCag{u}
  y\lambda(r)\in\closure{\mathcal{C}}$ implies
  $\labelledFormula{T}{\phi}{y\lambda(r)}\in\mathcal{F}$ 
\item If $\labelledFormula{F}{\BIDt{u}{r}\phi}{x}\in\mathcal{F}$, then
  $\exists y\in \Lambda_r$, $x\relCag{u}
  y\lambda(r)\in\closure{\mathcal{C}}$ and
  $\labelledFormula{F}{\phi}{y\lambda(r)}\in\mathcal{F}$ 
\item If $\labelledFormula{T}{\BIEt{u}{r}\phi}{x}\in\mathcal{F}$, then
  there exists $y \in \Lambda_r$, $x\lambda(r)\relCag{u}
  y\lambda(r)\in\closure{\mathcal{C}}$ and
  $\labelledFormula{T}{\phi}{y\lambda(r)}\in\mathcal{F}$ 
\item If $\labelledFormula{F}{\BIEt{u}{r}\phi}{x}\in\mathcal{F}$, then
  for all $y \in \Lambda_r$, $x\lambda(r)\relCag{u}
  y\lambda(r)\in\closure{\mathcal{C}}$ implies
  $\labelledFormula{F}{\phi}{y\lambda(r)}\in\mathcal{F}$. 
\end{enumerate}}
\label{definition_Hintikka}
\end{definition}

Conditions 1 to 4 ensure that a Hintikka CSS is not closed
and conditions 5 to 29 ensure that it is saturated (no new
tableaux rule can be applied). 

To extract countermodels, we must manipulate equivalence classes. 
The equivalence class of $x\in\domainR{\closure{\mathcal{C}}}$,
denoted $[x]$, is the set $ [x] = \{ y\in \Lambda_r\ |\ x\relCr
y\in\closure{\mathcal{C}} \} $. Moreover the function $\rho$ that
extracts a representative from a class is defined for any class $[x]$
by $\rho([x])=r$ if $\exists r\in Res/ \lambda(r)\in[x]$ and by
$\rho([x])=y$ with $y$ an arbitrary element of $[x]$ otherwise. 
We note that $\rho_x = \rho([x])$ and that 
$Rep(\domainR{\closure{\mathcal{C}}})$, the set of all representatives
of  $\domainR{\closure{\mathcal{C}}}$, is given by  
$Rep(\domainR{\closure{\mathcal{C}}}) = \{\rho_x\ |\
x\in\domainR{\closure{\mathcal{C}}}\} $.  

\begin{lemma}
For any set of constraints $\mathcal{C}$, we have $e \in
Rep(\domainR{\closure{\mathcal{C}}})$ and $\rho_\epsilon = e$. 

\label{lem_eq_class_epsilon}
\end{lemma}

\begin{definition}[Function $\Omega$]
Let $\CSS{\mathcal{F}}{\mathcal{C}}$ be a Hintikka CSS. The function
$\Omega$ associates to $\CSS{\mathcal{F}}{\mathcal{C}}$ a 3-tuple 
$\Omega(\CSS{\mathcal{F}}{\mathcal{C}}) = (\PRM, \{ \relR{a} \}_{a \in
\setAg}, \val)$, where $\PRM = (\setR, \compR)$, such that    
  \begin{itemize}
    \item $\setR = Rep(\domainR{\closure{\mathcal{C}}})\cup\ Res$, 
    \item  if $\alpha\notin Rep(\domainR{\closure{\mathcal{C}}})$ or
      $\beta\notin Rep(\domainR{\closure{\mathcal{C}}})$,
      then $\alpha\bullet\beta = \uparrow$,  
      else, $\alpha = \rho_x$ and $\beta=\rho_y$, and we have \\
      $\rho_x\compR \rho_y = 
    \left\{
    \begin{array}{ll}
        \uparrow    & \mbox{   if }
        xy \not \in \domainR{\closure{\mathcal{C}}} \\
        \rho_{xy} & \mbox{   otherwise,}
    \end{array}
    \right.$
    \item for all $a \in \setAg$,
          $\alpha \relR{a}\beta$
          iff $\alpha = \rho_x$ and $\beta=\rho_y$ and
          $x \relCag{a}  y \in \closure{\mathcal{C}}$, and 

    \item 
      $\alpha \in \val(p)$ iff $\alpha = \rho_x$ and 
      there exists $y \in \Lambda_r$ such that
      $y \relCr x \in \closure{\mathcal{C}}$
      and
      $\labelledFormula{T}{p}{y} \in \mathcal{F}$.
  \end{itemize}
  \label{def_function_omega}
\end{definition}

\begin{lemma}
Let $\CSS{\mathcal{F}}{\mathcal{C}}$ be a Hintikka
CSS. $\Omega(\CSS{\mathcal{F}}{\mathcal{C}})$ is a model.   
\label{lem_omega_extracts_model}
\end{lemma}

\begin{lemma}
Let $\CSS{\mathcal{F}}{\mathcal{C}}$ be a Hintikka CSS and $\model =
\Omega(\CSS{\mathcal{F}}{\mathcal{C}}) = (\PRM, \{ \relR{a} \}_{a \in
  \setAg}, \val)$, where $\PRM = (\setR, \compR)$. For any formula
$\phi \in \lang$, any agent $a \in \setAg$ and any $x,y \in
\domainR{\closure{\mathcal{C}}}$, we have: 
(1) If $\labelledFormula{F}{\phi}{x} \in \mathcal{F}$, then $\rho_{x}
\not \satisfies \phi$; (2) If $\labelledFormula{T}{\phi}{x} \in
\mathcal{F}$, then $\rho_{x} \satisfies \phi$. 
\label{lem_hintikka_forcing}
\end{lemma}

\begin{lemma}
Let $\CSS{\mathcal{F}}{\mathcal{C}}$ be a Hintikka CSS such that
$\labelledFormula{F}{\phi}{x} \in \mathcal{F}$. The formula $\phi$ is
not valid and $\Omega(\CSS{\mathcal{F}}{\mathcal{C}})$ is a
countermodel of $\phi$.   
\label{lem_hintikka_css_counter_model}
\end{lemma}

\begin{proof}
Let $\CSS{\mathcal{F}}{\mathcal{C}}$ be a Hintikka CSS such that
$\labelledFormula{F}{\phi}{x} \in \mathcal{F}$. Let $\mathcal{K} =
\Omega(\CSS{\mathcal{F}}{\mathcal{C}})$. By Lemma
\ref{lem_omega_extracts_model}, $\mathcal{K}$ is a model. As
$\CSS{\mathcal{F}}{\mathcal{C}}$  is a CSS, then by $(P_{css})$ and
Corollary \ref{cor_x_in_domain_EQUIV_x_relCr_x}, $x \in
\domainR{\closure{\mathcal{C}}}$. Thus, by Lemma
\ref{lem_hintikka_forcing}, we have $\rho_{x} \not \satisfies \phi$. 
Therefore, $\mathcal{K}$ is a countermodel of the formula $\phi$ and we
can conclude that $\phi$ is not valid.   
\end{proof}

\section{Proof of completeness} \label{sec:proof-of-completeness}

This proof is an extension of the proof for \BBI\ \cite{Lar14a} to the
epistemic connectives of our logic. It consists in identifying two things. First, a 
Hintikka CSS, using a fair strategy, from a formula for which there is no 
tableaux proof; that is, a sequence of labelled formulae in which all
labelled formulae occur infinitely many times. Second, an oracle; that is, 
a set of non-closed CSSs with some specific properties.  

\begin{definition}[Fair strategy]
A \emph{fair strategy} is a sequence of labelled formulae and agent constraints
$(S_i)_{i \in \mathbb{N}}$ in $(\{\mathbb{T}, \mathbb{F}\} \times
\lang  \times \Lambda_r) \cup (\Lambda_r  \times \setAg \times \Lambda_r)$ such that all
labelled formulae and all  agent constraints occur infinitely many
times in this sequence; that is, $\{i \in \mathbb{N} \mid S_i \equiv
\labelledFormula{S}{F}{x} \}$ and $\{i \in \mathbb{N} \mid S_i \equiv
x \lambda(r) \relCag{u} y \}$  are infinite, for any
$\labelledFormula{S}{F}{x} \in \{\mathbb{T}, \mathbb{F}\} \times \lang
\times \Lambda_r$ and any $x \lambda(r) \relCag{u} y \in \Lambda_r \times \setAg
\times \Lambda_r$.  
\label{def_fair_strategy}
\end{definition}

\begin{proposition}
There exists a fair strategy.
\label{prop_exists_fair_strategy}
\end{proposition}

\begin{proof}
Let $X = (\{\mathbb{T}, \mathbb{F}\} \times \lang \times \Lambda_r)
\cup (\Lambda_r \times \setAg \times \Lambda_r)$.
As $Prop$ is countable then $\lang$ is
countable. Moreover, $\Lambda_r$ is countable
(remember that $\gamma_r$ is countable).
Therefore, $X$ is countable. So $\mathbb{N}
\times X$ is countable and there exists a surjective
function $\varphi : \mathbb{N} \longrightarrow \mathbb{N}
\times X$.
Let $p : \mathbb{N} \times X \longrightarrow X$ defined by
$p(i,x) = x$  and $u = p \circ \varphi$. 
We show that $u$ is a fair strategy by showing that for any
$x \in X$, $u^{-1}(\{x\})$ is infinite.
Let $x \in X$. $u^{-1}(\{x\}) = \varphi^{-1}(p^{-1}(\{x\}))$. 
But $p^{-1}(\{x\}) = \{(i,x) | i \in \mathbb{N}\}$
so $p^{-1}(x)$ is infinite.  As $\varphi$ is surjective
$\varphi^{-1}(p^{-1}(\{x\}))$ is also infinite. 
\end{proof}

\begin{definition}
Let $\wp$ be a set of CSS.
\begin{enumerate}
  \item $\wp$ is \emph{$\CSSinclusion$-closed} if
        $\CSS{\mathcal{F}}{\mathcal{C}} \in \wp$
        holds whenever $\CSS{\mathcal{F}}{\mathcal{C}}
        \CSSinclusion \CSS{\mathcal{F}'}{\mathcal{C}'}$ and
        $\CSS{\mathcal{F}'}{\mathcal{C}'} \in \wp$
        holds. 
  \item $\wp$ is of \emph{finite character} if
        $\CSS{\mathcal{F}}{\mathcal{C}} \in \wp$
        holds whenever
        $\CSS{\mathcal{F}_f}{\mathcal{C}_f} \in
        \wp$ holds for every
        $\CSS{\mathcal{F}_f}{\mathcal{C}_f}
        \CSSfiniteInclusion \CSS{\mathcal{F}}{\mathcal{C}}$. 
  \item $\wp$ is \emph{saturated} if, for any
        $\CSS{\mathcal{F}}{\mathcal{C}} \in \wp$ and any instance
        \begin{prooftree}
          \AxiomC{$cond(\mathcal{F}, \mathcal{C})$}
          \UnaryInfC{$\CSS{\mathcal{F}_1}{\mathcal{C}_1} \ \mid \
            \ldots \ \mid \ \CSS{\mathcal{F}_k}{\mathcal{C}_k}$}    
        \end{prooftree} 
        of a rule of Figure \ref{fig_tableaux_method_rules},
        if $cond(\mathcal{F}, \mathcal{C})$ is fulfilled, then
        $\CSS{\mathcal{F} \cup \mathcal{F}_i}{\mathcal{C} \cup
        \mathcal{C}_i} \in \wp$ 
        for at least one $i \in \{ 1, \ldots, k \}$.
\end{enumerate}
\end{definition}

\begin{definition}[Oracle]
An \emph{oracle} is a set of non-closed CSSs that is
$\CSSinclusion$-closed, of finite character, and saturated. 
\label{def_oracle}
\end{definition}

\begin{lemma}
There exists an oracle which contains every finite CSS for which there
exists no closed tableau. 
\label{lem_oracle_exists}
\end{lemma}

\begin{proof}
The proof is an adaptation for our epistemic modalities of the
corresponding proof schema in \cite{Cour15a,Lar14a}. The proof given in
\cite{Cour15a} provides the necessary notions to develop this proof in
detail.   
\end{proof}

To prove completeness, we consider a formula $\varphi$ for which
there exists no proof and we show that there exists a countermodel for
this formula. 

The proof depends on finding a way to obtain a Hintikka CSS. By Lemma
\ref{lem_oracle_exists}, there exists an oracle which contains every
finite CSS for which there exists no closed tableau. We denote by $\wp$
this oracle. By Proposition \ref{prop_exists_fair_strategy}, there
exists a fair strategy. We denote by $\mathcal{S}$ this strategy and
$\mathcal{S}_i$ the $i^\text{th}$ formula or agent constraint of
$\mathcal{S}$. As $\mathcal{T}_0$ can not be closed then its unique
branch belongs to the oracle, that is 
$\CSS{\{ \labelledFormula{F}{\varphi}{c_1} \}}{\{ c_1 \relCr c_1 \}}
\in \wp$. 

We build a sequence $\CSS{\mathcal{F}_i}{\mathcal{C}_i}_{i \geqslant 0}$ 
whose limit is a Hintikka CSS, as follows:
\begin{itemize}
\item $\CSS{\mathcal{F}_0}{\mathcal{C}_0} =
      \CSS{\{ \labelledFormula{F}{\varphi}{c_1} \}}{\{ c_1 \relCr c_1  \}}$;  
\item $\mathcal{S}_i$ is a labelled formula of the form 
  $ \labelledFormula{S}{F}{x}$: 
  \begin{itemize}
    \item[-] If $\CSS{\mathcal{F}_i \cup \{
          \labelledFormula{S}{F}{x} \}
          }{\mathcal{C}_i} \not \in \wp$,  
          then  $\CSS{\mathcal{F}_{i + 1}}{\mathcal{C}_{i +
          1}} =
          \CSS{\mathcal{F}_i}{\mathcal{C}_i}$;   
    \item[-] If $\CSS{\mathcal{F}_i \cup \{
          \labelledFormula{S}{F}{x} \}
          }{\mathcal{C}_i} \in \wp$,  
          then 
          $\CSS{\mathcal{F}_{i + 1}}{\mathcal{C}_{i + 1}} =
          \CSS{\mathcal{F}_i \cup \{ \labelledFormula{S}{F}{x} \}
          \cup F_e }{\mathcal{C}_i \cup \mathcal{C}_e}$ such that $F_e$
          and $\mathcal{C}_e$ are given by \\
\small{    
\begin{center}
    {
    \renewcommand{\arraystretch}{1.4}
    \begin{tabular}{|c|c||c|c|}
    \hline
    $\mathbb{S}_i$ & $F_i$ & $F_e$ & $\mathcal{C}_e$  \\ 
    \hline \hline 
    $\mathbb{T}$ & $\BImI$ & $\emptyset$ & $\{ x \relCr \epsilon \}$ \\ 
    \hline 
    $\mathbb{T}$ & $\phi \BImast \psi$ & $\{
    \labelledFormula{T}{\phi}{\mathfrak{a}},
    \labelledFormula{T}{\psi}{\mathfrak{b}} \}$ & 
    $\{ x \relCr \mathfrak{a}\mathfrak{b} \}$\\  
    \hline 
    $\mathbb{F}$ & $\phi \BImimp \psi$ & $\{
    \labelledFormula{T}{\phi}{\mathfrak{a}},
    \labelledFormula{F}{\psi}{x \mathfrak{a}} \}$ 
    & $\{ x \mathfrak{a} \relCr x \mathfrak{a} \}$ \\  
    \hline 
    $\mathbb{F}$ & $\BIC{u}{r}\phi$ & $\{
    \labelledFormula{F}{\phi}{\mathfrak{a}} \}$ &
    $\{x\lambda(r)\relCag{u} \mathfrak{a}  \}$\\ 
    \hline 
    $\mathbb{T}$ & $\BID{u}{r}\phi$ & $\{
    \labelledFormula{T}{\phi}{\mathfrak{a}\lambda(r)} \}$ &
    $\{x\relCag{u} \mathfrak{a}\lambda(r)  \}$\\ 
    \hline 
    $\mathbb{F}$ & $\BIE{u}{r}\phi$ & $\{
    \labelledFormula{F}{\phi}{\mathfrak{a}\lambda(r)} \}$ &
    $\{x\lambda(r)\relCag{u} \mathfrak{a}\lambda(r)  \}$\\ 
    \hline  
    $\mathbb{T}$ & $\BICt{u}{r}\phi$ & $\{
     \labelledFormula{T}{\phi}{\mathfrak{a}} \}$ &
     $\{x\lambda(r)\relCag{u} \mathfrak{a}  \}$\\ 
    \hline 
    $\mathbb{F}$ & $\BIDt{u}{r}\phi$ & $\{
    \labelledFormula{F}{\phi}{\mathfrak{a}\lambda(r)} \}$ &
    $\{x\relCag{u} \mathfrak{a}\lambda(r)  \}$\\ 
    \hline 
    $\mathbb{T}$ & $\BIEt{u}{r}\phi$ & $\{
    \labelledFormula{T}{\phi}{\mathfrak{a}\lambda(r)} \}$ &
    $\{x\lambda(r)\relCag{u} \mathfrak{a}\lambda(r)  \}$\\ 
    \hline 
    \multicolumn{2}{|c||}{Otherwise} & $\emptyset$ & $\emptyset$ \\ 
    \hline  
    \end{tabular}
    }\\
    with $\mathfrak{a} = c_{2i+2}$ and $\mathfrak{b} = c_{2i+3}$. \\
    \end{center}
}
 \end{itemize}
      
\item $\mathcal{S}_i$ is an agent constraint of the form $x \lambda(r)
  \relCag{u} y$: 
  \begin{itemize}
    \item[-] If $\gamma_r \cap (\subLabelR{x} \cup \subLabelR{y}) \not
      \subseteq \{c_1, ..., c_{2i+1} \}$, 
          then $\CSS{\mathcal{F}_{i + 1}}{\mathcal{C}_{i + 1}} =
          \CSS{\mathcal{F}_i}{\mathcal{C}_i}$; 
    \item[-] If $\CSS{\mathcal{F}_i}{\mathcal{C}_i \cup \{ x \lambda(r),
        \relCag{u} y \} } \not \in \wp$  
          then $\CSS{\mathcal{F}_{i + 1}}{\mathcal{C}_{i + 1}} =
          \CSS{\mathcal{F}_i}{\mathcal{C}_i}$;   
    \item[-] If $\CSS{\mathcal{F}_i}{\mathcal{C}_i \cup \{ x \lambda(r)
        \relCag{u} y \} } \in \wp$,  then 
          $\CSS{\mathcal{F}_{i + 1}}{\mathcal{C}_{i + 1}} =
          \CSS{\mathcal{F}_i}{\mathcal{C}_i  \cup \{ x \lambda(r)
            \relCag{u} y \}}$.  
  \end{itemize}
\end{itemize}

\begin{proposition}
For any $i \in \mathbb{N}$, the following properties hold:
\begin{enumerate}
\item $\labelledFormula{F}{\varphi}{c_1} \in \mathcal{F}_i$
      and $c_1 \relCr c_1 \in \mathcal{C}_i$; 
\item $\mathcal{F}_i \subseteq \mathcal{F}_{i + 1}$ and
      $\mathcal{C}_i \subseteq \mathcal{C}_{i + 1}$; 
\item $\CSS{\mathcal{F}_i}{\mathcal{C}_i}_{i \geqslant 0} \in \wp$; 
\item $\alphabetR{\mathcal{C}_i} \subseteq \{ c_1, c_2, \ldots, c_{2i+1} \}$.
\end{enumerate}
\label{prop_constructed_branch}
\end{proposition}

The limit CSS $\CSS{\mathcal{F}_\infty}{\mathcal{C}_\infty}$ of
$\CSS{\mathcal{F}_i}{\mathcal{C}_i}_{i \geqslant 0}$ is defined 
by  
$  \mathcal{F}_\infty = \bigcup_{i \geqslant 0} \mathcal{F}_i$, 
$  \mathcal{C}_\infty = \bigcup_{i \geqslant 0} \mathcal{C}_i$.

\begin{proposition}
The following properties hold:
  \begin{enumerate}
  \item $\CSS{\mathcal{F}_\infty}{\mathcal{C}_\infty} \in \wp$; 
  \item For any labelled formula $\labelledFormula{S}{\phi}{x}$,
        if $\CSS{\mathcal{F}_\infty \cup \{
\labelledFormula{S}{\phi}{x} \}}{\mathcal{C}_\infty} \in \wp$, 
then $\labelledFormula{S}  {\phi}{x} \in \mathcal{F}_\infty$; 
\item For any agent constraint $x \lambda(r) \relCag{u} y$,
if $\CSS{\mathcal{F}_\infty}{\mathcal{C}_\infty \cup \{
x \lambda(r) \relCag{u}y \}} \in \wp$, 
then $x \lambda(r) \relCag{u} y \in \mathcal{C}_\infty$. 
\end{enumerate}  
\label{prop_limit_branch}
\end{proposition}

\begin{lemma}
The limit CSS is an Hintikka CSS.
\label{lem_limit_branch_is_hintikka_css}
\end{lemma}

\begin{proof}
By Proposition \ref{prop_limit_branch},
$\CSS{\mathcal{F}_\infty}{\mathcal{C}_\infty} \in \wp$. We
must verify that all conditions of Definition
\ref{definition_Hintikka} hold.
\end{proof}

\begin{theorem}[Completeness]
Let $\varphi$ be a formula. If $\varphi$ is valid, then there exists a
proof for $\varphi$.  
\label{th_completeness}
\end{theorem}

\begin{proof}
Similar to the proof of the corresponding result in \cite{CGP15}. 
We suppose that there is no proof for the formula $\varphi$ and show
that $\varphi$ is not valid. The method which we present here allows
us to build a limit CSS $\CSS{\mathcal{F}_\infty}{\mathcal{C}_\infty}$
that, by Lemma \ref{lem_limit_branch_is_hintikka_css}, is a Hintikka CSS. 
By property 1 of Proposition \ref{prop_constructed_branch}, 
$\labelledFormula{F}{\varphi}{c_1} \in \mathcal{F}_i$, for any $i \geqslant 0$. 
By the definition of a limit CSS, $\labelledFormula{F}{\varphi}{c_1} \in \mathcal{F}_\infty$. 
By Lemma \ref{lem_hintikka_css_counter_model},  $\varphi$ is not valid.  
\end{proof}

\end{document}